\documentclass[12pt]{article}

\usepackage{amsmath,amsthm,amssymb} % math
\usepackage{mathbbol} % used for some mathematical symbol, like \mathbb{\Lambda}
\usepackage{enumitem} % used for enumerations
\usepackage{graphicx} % used for figures

\usepackage{hyperref} % used to ref appendix

\usepackage[labelformat=simple]{subcaption} % used for subtables
\newtheorem{theorem}{Theorem}
\newtheorem{lemma}[theorem]{Lemma}
\newtheorem{corollary}[theorem]{Corollary}

\theoremstyle{definition}
\newtheorem{definition}[theorem]{Definition}

\newtheorem{remark}[theorem]{Remark}

\title{Fluid dynamics on logarithmic lattices}
\author{Ciro S. Campolina\thanks{Instituto de Matem\'{a}tica Pura e Aplicada -- IMPA, 22460-320 Rio de Janeiro, Brazil. E-mails: \texttt{sobrinho@impa.br, alexei@impa.br.}} \and Alexei A. Mailybaev\footnotemark[1]}

\date{}

\begin{document}
	
	\maketitle
	
	\begin{abstract}
		Open problems in fluid dynamics, such as the existence of finite-time singularities (blowup), explanation of intermittency in developed turbulence, etc., are related to multi-scale structure and symmetries of underlying equations of motion.
		Significantly simplified equations of motion, called toy-models, are traditionally employed in the analysis of such complex systems. In such models, equations are modified preserving just a part of the structure believed to be important. Here we propose a different approach for constructing simplified models, in which instead of simplifying equations one introduces a simplified configuration space: velocity fields are defined on  multi-dimensional logarithmic lattices with proper algebraic operations and calculus. Then, the equations of motion retain their \textit{exact original form} and, therefore, naturally maintain most scaling properties, symmetries and invariants of the original systems. Classification of such models reveals a fascinating relation with renowned mathematical constants such as the golden mean and the plastic number. Using both rigorous and numerical analysis, we describe various properties of solutions in these models, from the basic concepts of existence and uniqueness to the blowup development and turbulent dynamics. In particular, we observe strong robustness of the chaotic blowup scenario in the three-dimensional incompressible Euler equations, as well as the Fourier mode statistics of developed turbulence resembling the full three-dimensional Navier-Stokes system.
	\end{abstract}
	
	\section{Introduction}
	
	The theory of multi-scale nonlinear flows and, in particular, the phenomenon of hydrodynamic turbulence comprise a multitude of yet unresolved problems: the global regularity~\cite{fefferman2006existence} and existence of finite-time singularities~\cite{gibbon2008three2}, explanation of intermittency~\cite{frisch1999turbulence} and dissipation anomaly~\cite{eyink2006onsager}, to name a few.
	Many of these problems determine the state-of-the-art in nonlinear science and open new areas in mathematics and physics.
	In these studies, toy-models employed as caricatures of complex phenomena have been proved to be indispensable as the testing ground for new ideas and theories.
	Such models retain some basic features believed to be important, while  the remaining content is simplified as much as possible.
	The conventional simplifications are related to reducing the spatial dimension, e.g., the one-dimensional Burgers equation~\cite{burgers1948mathematical} or the Constantin-Lax-Majda model~\cite{constantin1985simple} with further generalizations~\cite{okamoto2008generalization,choi2017finite}.
	The number of degrees of freedom can be drastically decreased by exploring the cascade ideas in the so-called shell models of turbulence~\cite{biferale2003shell}.
	In these models, multi-scale properties are mimicked by geometrical progressions of scales, resulting in the popular GOY~\cite{gledzer1973system,ohkitani1989temporal} and Sabra models~\cite{l1998improved}, the reduced wave vector set approximation (REWA)~\cite{eggers1991does,grossmann1996developed} and tree models~\cite{benzi1993random,benzi19971}, as well as more sophisticated geometric constructions~\cite{gurcan2017nested,gurcan2019spiral}. 
	Toy-models rely on the intuitive decision of what unimportant properties of the original system can be neglected.
	Of course, dealing with open problems, such decision has the risk of missing essential features of motion.
	Especially, this concerns neglected symmetries and conserved quantities, since fluid systems are known to possess highly nontrivial (infinite dimensional) symmetry groups and conservation laws~\cite{zakharov1997hamiltonian,majda2002vorticity}, e.g., the Kelvin Circulation Theorem.
	
	In the present work, we propose a different approach for constructing simplified models, in which instead of simplified equations one introduces a simplified configuration space with proper algebraic operations and calculus.
	For this purpose, we employ velocity fields defined on discrete multi-dimensional lattices with logarithmically distributed nodes.
	These lattices are designed such that the equations of motion can be used in their exact original form and, as a consequence, the symmetry groups and conservation laws automatically carry over to the new system.
	The resulting models possess much higher degree of similarity to the exact equations as compared to conventional toy-models and, at the same time, share the property of being easily accessible for numerical analysis.
	
	The paper is divided logically into two parts.
	The first part consists of Sections~\ref{SEC:loglatt}--\ref{SEC:generalized}.
	Here we classify logarithmic lattices on which the functional operations with necessary properties can be introduced.
	This classification reveals interesting connections with well-known mathematical constants, associating the two representative lattice spacings with the golden mean and the plastic number.
	We prove that the product of two fields cannot be associative on logarithmic lattices, but has the property of associativity in average.
	With this limitation, the technique is applicable to any system of partial differential equations with quadratic nonlinearities and quadratic or linear conservation laws or other integral characteristics.
	Fortunately, this is sufficient for the applicability of our approach to many fundamental models in fluid dynamics.
	
	The second part includes Sections~\ref{SEC:ideal_flow},~\ref{SEC:viscous_flow} with the Appendix, and contains applications of the developed technique to specific equations of fluid dynamics.
	Here many properties of solutions are proved following classical derivations in fluid dynamics, as a consequence of the designed similarity of configuration spaces.
	This refers not only to the basic symmetries and inviscid invariants like, e.g., energy and helicity, but also to a number of fine properties such as incompressibility and conservation of circulation.
	In this way, we demonstrate basic properties such as local-in-time existence and uniqueness of strong solutions and the blowup criterion.
	Then, we proceed with the numerical study.
	Our central numerical result concerns the blowup problem for incompressible 3D Euler equations, where we demonstrate surprisingly strong robustness of the chaotic blowup scenario~\cite{campolina2018chaotic} with respect to the choice of the logarithmic lattice.
	We also show that, for the 3D Navier-Stokes equations at large Reynolds numbers, our models demonstrate chaotic regimes with the classical properties of developed turbulence.
	Here, velocity fields on logarithmic lattices must be considered as analogues of Fourier-transposed velocity fields in the original model. 
	
	The paper is organized as follows.
	We classify logarithmic lattices in Section~\ref{SEC:loglatt}, introduce the calculus on such lattices in Section~\ref{SEC:calculus}, and provide some generalizations in Section~\ref{SEC:generalized}.
	Section~\ref{SEC:ideal_flow} is devoted to the blowup problem in incompressible ideal flows, and Section~\ref{SEC:viscous_flow} to the numerical study of developed turbulence.
	We draw  conclusions in the final section.
	Appendix discusses the Burgers equation, where connections to existing shell models are given, as well as some perspectives for compressible flows.
	
	\section{Logarithmic lattices}
	\label{SEC:loglatt}
	
	\begin{figure*}[t]
		\centering
		\includegraphics[width=.65\textwidth]{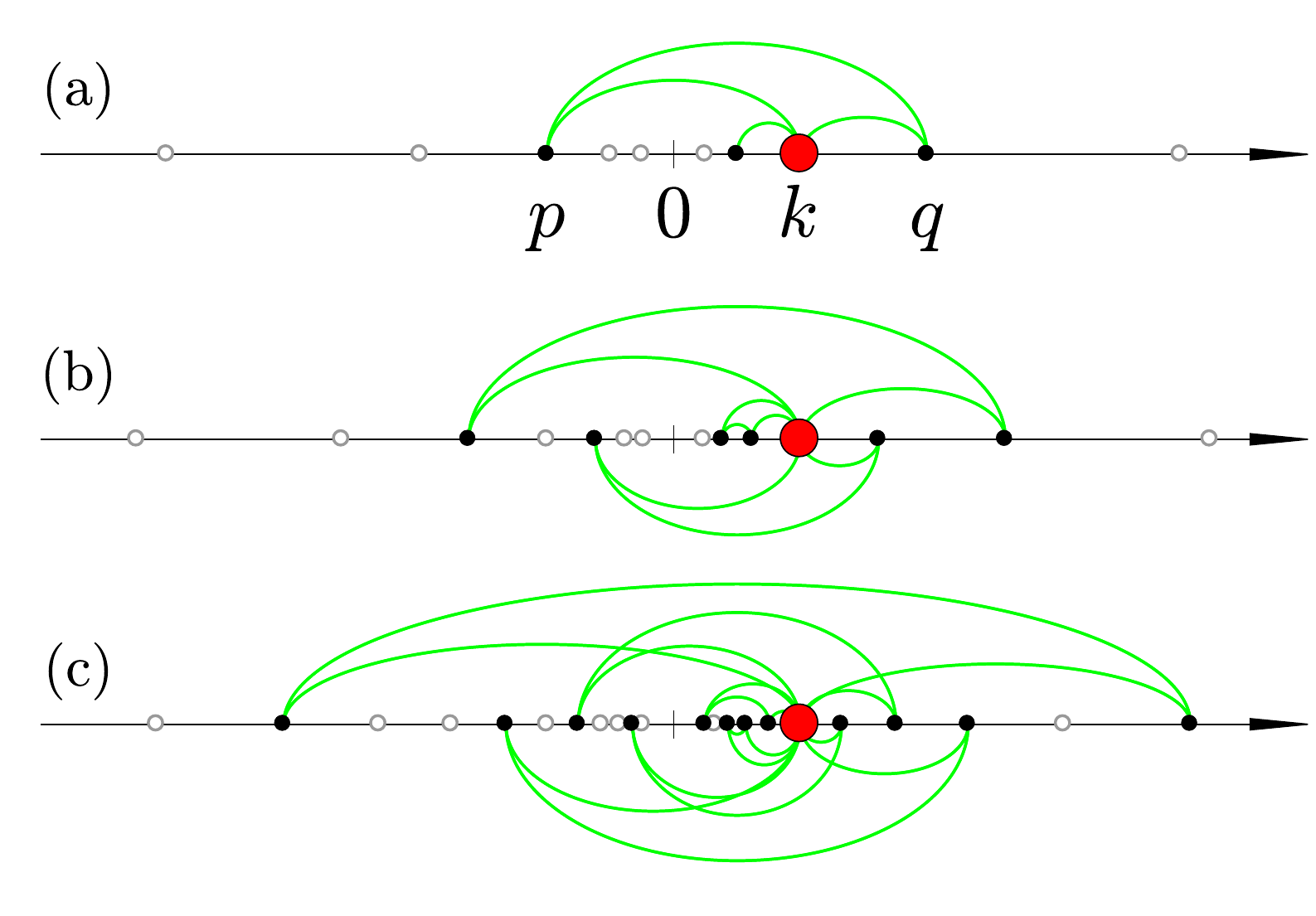}
		\caption{Triad interactions on logarithmic lattices with different spacing factors: (a) $\lambda = 2$; (b) $\lambda = \varphi$, the golden mean; (c) $\lambda = \sigma$, the plastic number. The red node $k \in \mathbb{\Lambda}$ can be decomposed into sums $k = p + q$, where all possible $p,q \in \mathbb{\Lambda}$ are shown by the green lines. All figures are given in the same scale.}
		\label{FIG:triads}
	\end{figure*}
	
	In this section, we perform a systematic study of logarithmic lattices with certain geometric properties, providing the domain on which the dynamical models shall be defined in the next sections. We start with one-dimensional lattices, similar to those used in shell models, and then consider the multi-dimensional case. 
	
	Given a real number $\lambda>1$, the \textit{logarithmic lattice} with spacing factor $\lambda$ is the set
	\begin{equation}
	\mathbb{\Lambda} = \{ \pm \lambda^n\}_{n \in \mathbb{Z}},
	\end{equation}
	consisting of positive and negative integer powers of $\lambda$ -- see Fig.~\ref{FIG:triads}.
	This set has two properties important for applications.
	First, $\mathbb{\Lambda}$ is scale-invariant, i.e., $\mathbb{\Lambda} = k \mathbb{\Lambda}$ for any $k \in \mathbb{\Lambda}$. 
	Secondly, the points of the lattice grow geometrically with $n$. 
	Thus, with only a few nodes we span a large range of scales.
	However, logarithmic lattices are not closed under addition as $p + q \notin \mathbb{\Lambda}$ for general $p,q \in \mathbb{\Lambda}$.
	Three points $k,p,q \in \mathbb{\Lambda}$ on a logarithmic lattice form a \textit{triad} if $k = p + q$. 
	In this case, we say that $k$ \textit{interacts} with $p$ and $q$.
	The lattice is called \textit{nondegenerate} if every two nodes interact through a finite sequence of triads.
	We are interested in a twofold task:
	\begin{enumerate}[label=(\roman*)]
		\item to determine which spacings $\lambda$ provide nondegenerate lattices, and
		\item to classify all triads of nondegenerate lattices.
	\end{enumerate}
	Because of the scale invariance, it is sufficient to describe the triads at unity, i.e., $1 = p + q$.
	
	\begin{table}[t]
		\begin{subtable}{.35\textwidth}
			\caption{} % title of Table
			\centering % used for centering table
			\begin{tabular}{c c c c} % centered columns
				\hline\hline %inserts double horizontal lines
				$i$&$\phantom{-}1$&$\phantom{-}2$&$3$ \\
				\hline
				$p_i$&$\phantom{-}2$&$-1$&$1/2$ \\ 
				$q_i$&$-1$&$\phantom{-}2$&$1/2$ \\
				\hline\hline
			\end{tabular}
			\label{TAB:DN}
		\end{subtable}\hspace*{.05\textwidth}
		\begin{subtable}{.6\textwidth}
			\caption{} % title of Table
			\centering % used for centering table
			\begin{tabular}{c c c c c c c} % centered columns 
				\hline\hline %inserts double horizontal lines
				$i$&$\phantom{-}1$&$\phantom{-}2$&$3$&$4$&$5$&$6$ \\
				\hline
				$p_i$&$\phantom{-}\lambda^b$&$-\lambda^a$&$\phantom{-}\lambda^{b-a}$&$-\lambda^{-a}$&$\lambda^{-b}$&$\lambda^{a-b}$ \\ 
				$q_i$&$-\lambda^a$&$\phantom{-}\lambda^b$&$-\lambda^{-a}$&$\phantom{-}\lambda^{b-a}$&$\lambda^{a-b}$&$\lambda^{-b}$ \\
				\hline\hline
			\end{tabular}
			\label{TAB:general}
		\end{subtable}\vspace*{20pt}
		\begin{subtable}{\textwidth}
			\caption{} % title of Table
			\centering % used for centering table
			\begin{tabular}{c c c c c c c c c c c c c} % centered columns 
				\hline\hline %inserts double horizontal lines
				$i$&$\phantom{-}1$&$\phantom{-}2$&$3$&$4$&$5$&$6$&$\phantom{-}7$&$\phantom{-}8$&$9$&$10$&$11$&$12$ \\
				\hline
				$p_i$&$\phantom{-}\sigma^3$&$-\sigma$&$\sigma^2$&$-\sigma^{-1}$&$\sigma^{-3}$&$\sigma^{-2}$&
				$\phantom{-}\sigma^5$&$-\sigma^4$&$\phantom{-}\sigma^{\phantom{-}}$&$-\sigma^{-4}$&$\sigma^{-5}$&$\sigma^{-1}$ \\
				$q_i$&$-\sigma^{\phantom{-}}$&$\phantom{-}\sigma^3$&$-\sigma^{-1}$&$\sigma^2$&$\sigma^{-2}$&$\sigma^{-3}$&
				$-\sigma^4$&$\phantom{-}\sigma^5$&$-\sigma^{-4}$&$\sigma$&$\sigma^{-1}$&$\sigma^{-5}$ \\
				\hline\hline
			\end{tabular}
			\label{TAB:extended}
		\end{subtable}
		\caption{Triads at the unity $1 = p_i + q_i$ for different spacing factors: \subref{TAB:DN}~$\lambda = 2$;
			\subref{TAB:general}~$\lambda$ satisfies $1 = \lambda^b-\lambda^a$ for integers $0 \leq a <b$. For example, $\lambda = \varphi$ is the golden mean for $a = 1$ and $b = 2$;
			\subref{TAB:extended}~$\lambda = \sigma$, the plastic number.}
		\label{TAB:triads_unity}
	\end{table}
	
	Lattices $\mathbb{\Lambda}$ with nontrivial triad interactions exist only for certain values of $\lambda$.
	Let us first present three specific nondegenerate lattices.
	The lattice with $\lambda = 2$ has three possible types of triads described in Tab.~\ref{TAB:DN} and Fig.~\ref{FIG:triads}(a).
	For any $k \in \mathbb{\Lambda}$, these triads are $k = \lambda k - k$, $k = -k + \lambda k$ and $k = \lambda^{-1}k + \lambda^{-1}k$.
	The next example is $\lambda = \varphi$, where $\varphi = (1+\sqrt{5})/2 \approx 1.618$ is the \textit{golden mean}.
	All triads are obtained from permutations and rescalings of the identity $1 = \varphi^2 - \varphi$, providing the richer sample of interactions in Tab.~\ref{TAB:general}.
	In this case, each point interacts with six different neighbors -- see Fig.~\ref{FIG:triads}(b).
	Another example is provided by the \textit{plastic number} of Dom Van der Laan~\cite{laan1960nombre}
	\begin{equation}
	\sigma = \frac{\sqrt[3]{9+\sqrt{69}}+\sqrt[3]{9-\sqrt{69}}}{\sqrt[3]{18}} \approx 1.325,
	\label{plastic}
	\end{equation}
	which is the common real solution of equations $\sigma^3 - \sigma - 1 = 0$ and $\sigma^5 - \sigma^4 - 1 = 0$.
	The lattice with spacing $\lambda = \sigma$ has twelve types of interacting triads, enumerated in Tab.~\ref{TAB:extended} and depicted in Fig. \ref{FIG:triads}(c).
	Because immediate neighbors are coupled, these are examples of nondegenerate lattices.
	On the other hand, if $\lambda = \sqrt{2}$, the lattice is degenerate: there are no interactions that couple points $\pm 2^n$ with $\pm 2^n\sqrt{2}$.
	
	The main result of this section is the classification of nondegenerate logarithmic lattices with respect to their triad interactions, given by the following
	
	\begin{theorem}
		The following three cases describe all nondegenerate lattices with spacing factors $\lambda \geq 1.05$:
		\begin{enumerate}[label=(\roman*)]
			\item $\lambda = 2$, and all triads at the unity are given in Tab. \ref{TAB:DN};
			\item $\lambda = \sigma$, the plastic number~\eqref{plastic}, and all triads at the unity are given in Tab. \ref{TAB:extended};
			\item $\lambda$ satisfies $1 = \lambda^b - \lambda^a$, where $(a,b)$ are mutually prime integers not larger than $62$, excluding also the pairs~$(a,b) = (1,3)$ and~$(4,5)$. All triads at the unity are given in Tab.~\ref{TAB:general}.
		\end{enumerate}
		\label{THE:all_triads}
	\end{theorem}
	
	\begin{remark}
		We used the lower bound $\lambda \geq 1.05$ in order to make the numerically assisted proof possible.
		Still, we conjecture that Theorem~\ref{THE:all_triads} is valid for arbitrary $\lambda > 1$,  with no upper bound for $a$ and $b$ in the item ($iii$).
		A partial result in this direction is the Theorem proved in~\cite{aarts2001morphic}, which states that the plastic number is the only common root greater than unity of any two distinct polynomials $\lambda^a - \lambda^{a-1} - 1$ and $\lambda^b - \lambda - 1$ with $a,b \geq 2$.
	\end{remark}
	
	\begin{proof}
		Let us consider the trinomial equation
		\begin{equation}
		p_{a,b}(\lambda) = \lambda^b - \lambda^a - 1 = 0,
		\label{EQ:trinomial}
		\end{equation}
		with integer powers $0 \leq a < b$.
		This equation has a single root in the interval $\lambda > 1$ because the function $p_{a,b}(\lambda)$ is strictly increasing in $\lambda \in [1,\infty)$ with image $[-1,\infty)$.
		Relation \eqref{EQ:trinomial} yields the three equalities
		\begin{equation}
		1 = \lambda^b - \lambda^a = \lambda^{b-a} - \lambda^{-a} = \lambda^{-b} + \lambda^{a-b}.
		\label{EQ:trinomial_rescaled}
		\end{equation}
		There are six triads $1 = p + q$ corresponding to expressions~\eqref{EQ:trinomial_rescaled} as described in Tab.~\ref{TAB:general}.
		Let us show that the lattice is degenerate when $a$ and $b$ have a common divisor $m > 1$.
		For the sublattice $\mathbb{\Lambda}^\prime = \{ \pm \lambda^{mn} \}_{n \in \mathbb{Z}}$ to be coupled with the remaining points, the spacing $\lambda$ should satisfy another trinomial equation~\eqref{EQ:trinomial} with exponents $(a',b')$ not multiples of $m$. However, this is not possible, as it follows from case (b) of Lemma~\ref{LEM:solution_trinomials} below.
		This leaves only the mutually prime pairs $(a,b)$ to our consideration.
		Now, the Theorem is a direct consequence of Lemma~\ref{LEM:solution_trinomials}, where all triads are generated by the relations in~\eqref{EQ:trinomial_rescaled}: the case~\textit{(i)} corresponds to $(a,b) = (0,1)$; the case~\textit{(ii)} to $(a,b) = (1,3)$ and $(4,5)$; and the case~\textit{(iii)} to all other possibilities.
	\end{proof}
	
	\begin{lemma}
		Consider two distinct trinomials~\eqref{EQ:trinomial} with integer powers $(a_1,b_1)$ and $(a_2,b_2)$, where $0 \leq a_1 < a_2$.
		These trinomials have a common root $\lambda \geq 1.05$ if and only if
		\begin{enumerate}[label=(\alph*)]
			\item $\lambda = \sigma$ is the plastic number~\eqref{plastic}. In this case, $(a_1,b_1) = (1,3)$ and $(a_2,b_2) = (4,5)$ are mutually prime;
			\item $\lambda = \sigma^{1/m}$ with $m = 2,\dots,5$. In this case $(a_1,b_1) = (m,3m)$ and $(a_2,b_2) = (4m,5m)$ have the same common divisor $m$.
		\end{enumerate}
		\label{LEM:solution_trinomials}
	\end{lemma}
	\begin{proof}
		Let us denote by $\lambda(a,b)$ the unique root of~\eqref{EQ:trinomial} in the interval $\lambda >1$.
		Note that $\lambda(a,b) < \lambda(a',b)$ if $0 \leq a' < a$ because the polynomials $p_{a,b}(\lambda)$ are strictly increasing starting at $p_{a,b}(1) = p_{a',b}(1) = -1$ and $p_{a,b}(\lambda) < p_{a',b}(\lambda)$ for $\lambda >1$.
		Therefore, if we fix the exponent~$b$ of trinomial $p_{a,b}(\lambda)$, then $\lambda(a,b)$ is maximized when $a = b-1$.
		Next, $\lambda(b-1,b)$ form a decreasing sequence with respect to $b$, since $p_{b-1,b}(\lambda) < p_{b,b+1}(\lambda)$.
		Finally, one may check that $p_{62,63}(1.05) > 0$, so $\lambda(62,63)<1.05$.
		Therefore $\lambda(a,b) \geq 1.05$ only if $b < 63$.
		This bound leaves a finite number of trinomials to our consideration. 
		Since the plastic number $\sigma$ satisfies $\sigma^5 - \sigma^4 - 1 = \sigma^3 - \sigma - 1 = 0$, we obtain the two cases~\textit{(a)} and~\textit{(b)} of the Lemma.
		It remains to check that trinomials with different powers have no common root.
		This was accomplished via Validated Numerics~\cite{moore2009introduction}, a computer assisted proof using the following strategy.
		Given two trinomials $p_{a,b}$ and $p_{a',b'}$, we estimate their respective roots $\lambda_1$ and $\lambda_2$ with Newton's Method up to machine double precision. Next, using Symbolic Algebra~\cite{cohen2003computer}, we evaluate exactly the product $p_{a,b}p_{a',b'}$ at the middle point $\lambda_m = (\lambda_1 + \lambda_2)/2$ of their approximate roots. For all cases, it was verified a negative number at this point, which guarantees that $\lambda(a,b) \neq \lambda(a',b')$.
	\end{proof}
	
	\begin{figure*}[t]
		\centering
		\includegraphics[width=\textwidth]{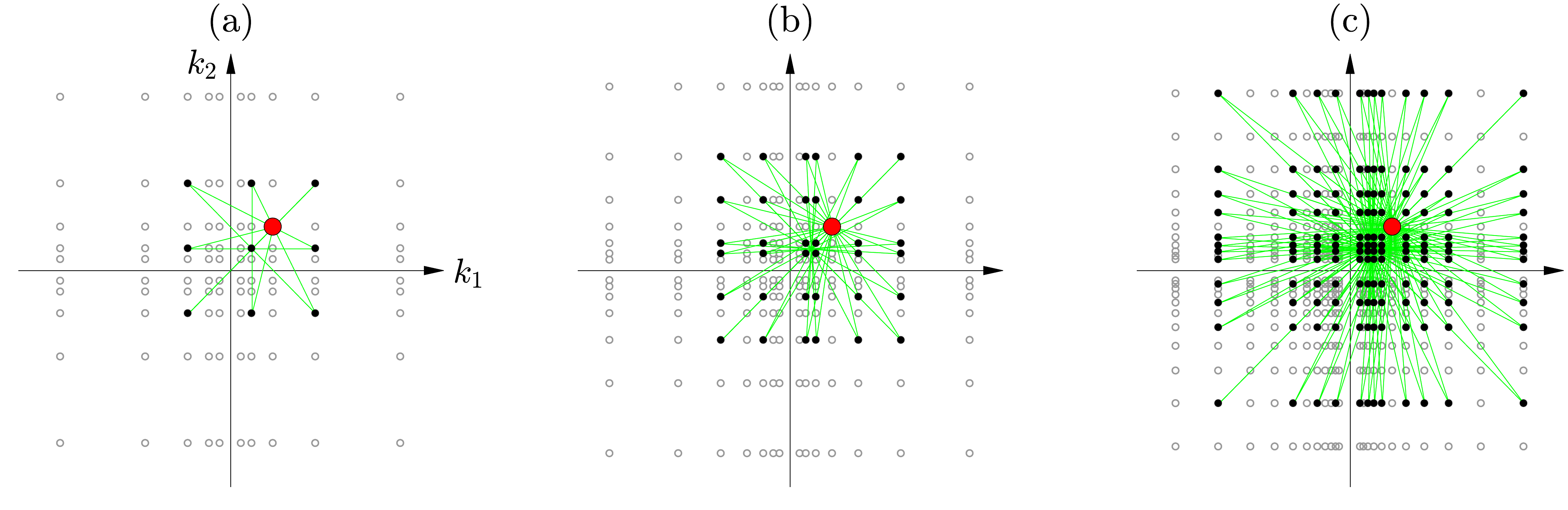}
		\caption{Triad interactions on two-dimensional logarithmic lattices for different spacing factors: (a) $\lambda = 2$; (b) $\lambda = \varphi$, the golden mean; (c) $\lambda = \sigma$, the plastic number. The red node $\mathbf{k}$ can be decomposed into sums $\mathbf{k} = \mathbf{p} + \mathbf{q}$ where all possible nodes $\mathbf{p}$ and $\mathbf{q}$ are indicated by the green lines. All figures are given in the same scale. From (a) to (c), both the density of nodes and the number of triads per each node increase.}
		\label{FIG:2D_lattice}
	\end{figure*}
	
	The above results for one-dimensional logarithmic lattices can be extended to higher dimensions.
	The \textit{$d$-dimensional logarithmic lattice} with spacing $\lambda>1$ is given by the cartesian power $\mathbb{\Lambda}^d = \mathbb{\Lambda} \times \cdots \times \mathbb{\Lambda}$ (with $d$ factors), i.e., $\mathbf{k} = (k_1,\dots,k_d) \in \mathbb{\Lambda}^d$ if each component $k_j \in \mathbb{\Lambda}$.
	Three points $\mathbf{k},\mathbf{p},\mathbf{q} \in \mathbb{\Lambda}^d$ on the lattice form a triad if $\mathbf{k} = \mathbf{p} + \mathbf{q}$.
	All nondegenerate lattices $\mathbb{\Lambda}^d$ are given by the spacings $\lambda$ listed in Theorem~\ref{THE:all_triads} and all triads are combinations of the one-dimensional triads for each component -- see Fig.~\ref{FIG:2D_lattice} for the two-dimensional picture.
	
	\section{Calculus on logarithmic lattices}
	\label{SEC:calculus}
	
	Let us consider complex-valued functions $f(\mathbf{k}) \in \mathbb{C}$ on a nondegenerate logarithmic lattice $\mathbb{\Lambda}^d$, where $\mathbf{k} \in \mathbb{\Lambda}^d$ is interpreted as a wave vector in Fourier space.
	Motivated by the property of the Fourier transform of a real-valued function, we impose the reality condition
	\begin{equation}
	f(-\mathbf{k}) = \overline{f(\mathbf{k})},
	\label{EQ:reality}
	\end{equation}
	where the bar denotes complex conjugation.
	Thus, $f(\mathbf{k})$ is analogous to the Fourier transform of a real function, and now we are going to introduce basic operations.
	
	Functions $f(\mathbf{k})$ possess a natural structure of a linear space with real scalars.
	Since we are working with Fourier-space representation, the spatial derivative $\partial_j$ in the $j$-th direction is defined by the Fourier factor,
	\begin{equation}
	\partial_j f(\mathbf{k}) = ik_j f(\mathbf{k}), \quad j =1,\dots,d,
	\label{EQ:fourier_factors}
	\end{equation}
	where $i$ is the imaginary unit.
	Clearly, higher order derivatives are products of such Fourier factors.
	Given two functions $f$ and $g$, one defines their \textit{inner product} naturally as
	\begin{equation}
	(f,g) = \sum_{\mathbf{k} \in \mathbb{\Lambda}^d} f(\mathbf{k})\overline{g(\mathbf{k})}.
	\label{DEF:inner}
	\end{equation}
	Just like the $L^2$-inner product of real functions, expression~\eqref{DEF:inner} is real valued because of reality condition~\eqref{EQ:reality}.
	
	The notion of differentiability on the lattice retains some important calculus identities, like the integration by parts
	\begin{equation}
	(\partial_j f, g) = - (f, \partial_j g), \quad j =1,\dots,d,
	\label{EQ:by_parts}
	\end{equation}
	which follows from the fact that the inner product \eqref{DEF:inner} couples $f(\mathbf{k})$ and $\overline{g(\mathbf{k})} = g(-\mathbf{k})$.
	We next define the product of two functions on the logarithmic lattice, which in Fourier space is understood as a convolution.
	Here and below, all functions are assumed to be absolutely summable
	\begin{equation}
	\label{EQ:L1}
	\sum_{\mathbf{k} \in \mathbb{\Lambda}^d} |f(\mathbf{k})| < \infty.
	\end{equation}
	
	\begin{definition}
		A \textit{product on the logarithmic lattice} $\mathbb{\Lambda}^d$, denoted by $\ast$, is a binary operation between absolutely summable functions on $\mathbb{\Lambda}^d$, which satisfies the following properties:
		\begin{enumerate}[label=\textit{(P.\arabic*)}] %\roman*
			\item \emph{(Reality condition)}
			$
			\!
			\begin{aligned}[t]
			(f \ast g)(-\mathbf{k}) = \overline{(f \ast g) (\mathbf{k})};
			\end{aligned}
			$ 
			\label{DEF:PROD_reality}
			\item \emph{(Bilinearity)} 
			$
			\!
			\begin{aligned}[t]
			(f + \gamma g) \ast h = f \ast h + \gamma(g \ast h),
			\end{aligned}
			$
			for any $\gamma \in \mathbb{R}$;
			\label{DEF:PROD_bilinearity}
			\item \emph{(Commutativity)}
			$
			\!
			\begin{aligned}[t]
			f \ast g = g \ast f;
			\end{aligned}
			$ 
			\label{DEF:PROD_commutativity}
			\item \emph{(Associativity in average)}
			$
			\!
			\begin{aligned}[t]
			(f \ast g, h ) = (f, g \ast h);
			\end{aligned}
			$
			\label{DEF:PROD_associativity_avg}
			\item \emph{(Leibniz rule)}
			$
			\!
			\begin{aligned}[t]
			\partial_j (f \ast g) = \partial_j f \ast g + f \ast \partial_j g, \ \text{for } j = 1,\dots,d;
			\end{aligned}
			$
			\label{DEF:PROD_leibniz}
		\end{enumerate}
		Additional properties, which are related to the spatial symmetries of the lattice, may be imposed:
		\begin{enumerate}[label=\textit{(S.\arabic*)}] %\roman*
			\item \emph{(Scaling invariance)}
			$
			\!
			\begin{aligned}[t]
			\delta_{\lambda}(f \ast g) = \delta_{\lambda} f \ast \delta_{\lambda} g,
			\end{aligned}
			$
			where we denoted $\delta_\lambda f (\mathbf{k}) = f(\lambda \mathbf{k})$, the rescaling of $f$ by the lattice spacing $\lambda$;
			\label{DEF:PROD_scaling}
			\item \emph{(Isotropy and parity)} 
			$
			\!
			\begin{aligned}[t]
			(f\ast g) \circ R = (f \circ R) \ast (g \circ R),
			\end{aligned}
			$
			where we denoted $(f \circ R)(\mathbf{k}) = f(R\mathbf{k})$ and $R \in \mathsf{O_h}$ is any element of the group of cube symmetries; cf.~\cite[Sec.~93]{landau2013quantum} -- it includes all transformations $(k_1,\dots,k_d) \mapsto (\pm k_{\alpha_1},\dots,\pm k_{\alpha_d})$, where $(\alpha_1,\dots,\alpha_d)$ are permutations of $(1,\dots,d)$.
			\label{DEF:PROD_isotropy}
		\end{enumerate}
		\label{DEF:PROD}
	\end{definition}
	\begin{remark}
		Lebniz rule readily implies \textit{translation invariance} on the lattice, expressed as $\tau_{\pmb{\xi}} (f \ast g) = \tau_{\pmb{\xi}} f \ast \tau_{\pmb{\xi}} g$, where $\tau_{\pmb{\xi}} f (\pmb{k}) = e^{-i\pmb{k} \cdot \pmb{\xi}} f(\pmb{k})$ mimics the physical-space translation (in Fourier representation) by any vector $\pmb{\xi} \in \mathbb{R}^d$;
	\end{remark}
	
	The required properties for the product are chosen in order to mimic as much as possible a common pointwise product (or, equivalently, a convolution in Fourier space) of real functions defined in the Euclidean space.
	The symmetries of scaling invariance \ref{DEF:PROD_scaling} and isotropy \ref{DEF:PROD_isotropy} can only be satisfied in a discrete form, because only discrete scalings and rotations are symmetries of the lattice itself.
	More importantly, we will prove shortly that the product cannot be associative.
	Nevertheless, the weaker property of associativity in average \ref{DEF:PROD_associativity_avg} can be satisfied, which turns out to be sufficient for our purposes.
	
	We first establish the general form of the product on one-dimensional lattices.
	Later, it will be generalized to higher dimensions.
	Bilinearity \ref{DEF:PROD_bilinearity}, Leibniz rule \ref{DEF:PROD_leibniz} and scaling invariance  \ref{DEF:PROD_scaling} yield the following general form of the product
	\begin{equation}
	(f \ast g)(k) = \sum_{p_j + q_j = 1} c_j f(p_j k)g(q_j k), \quad k \in \mathbb{\Lambda}.
	\label{EQ:general_product}
	\end{equation}
	Here, the Leibniz rule restricts the product to triad interactions, which are determined by the factors $p_j$ and $q_j$ from Tab.~\ref{TAB:triads_unity} for each lattice of Theorem~\ref{THE:all_triads}.
	The independence of the coefficients $c_j$ on $k$ is a consequence of the scaling invariance. Next, reality condition~\ref{DEF:PROD_reality} and parity $k \mapsto -k$, from~\ref{DEF:PROD_isotropy}, imply that the coefficients $c_j$ are real.
	Since the sum in~\eqref{EQ:general_product} has a finite number of terms, the product of two absolutely summable functions is absolutely summable.
	
	As an example, consider the case $\lambda = 2$.
	Then, for the three triads in Tab.~\ref{TAB:DN}, formula~\eqref{EQ:general_product} becomes
	\begin{equation}
	(f \ast g)(k) = c_1 f(2 k)g(-k) + c_2 f(-k)g(2 k) + c_3 f(2^{-1}k)g(2^{-1}k).
	\label{EQ:product_DN_general}
	\end{equation}
	We are interested in non-trivial products~\eqref{EQ:general_product}, where the coefficients $c_j$ do not vanish simultaneously.
	
	\begin{theorem}
		Let $\mathbb{\Lambda}$ be one of the logarithmic lattices~\textit{(i)--(iii)} described in Theorem~\ref{THE:all_triads}. 
		For the lattices~\textit{(i)} and~\textit{(iii)}, the product with properties~\ref{DEF:PROD_reality}--\ref{DEF:PROD_leibniz} and symmetries~\ref{DEF:PROD_scaling} and~\ref{DEF:PROD_isotropy} is unique, up to a real prefactor which we set to unity, and has the form
		\begin{equation}
		(f \ast g)(k) = \sum_{p_j + q_j = 1} f(p_j k) g(q_j k), \quad k \in \mathbb{\Lambda},
		\label{EQ:product_1D_1}
		\end{equation}
		where the coupling factors $p_j$ and $q_j$ are given in Tab.~\ref{TAB:triads_unity}.
		For the lattice~\textit{(ii)}, the general form of the product is
		\begin{equation}
		(f \ast g)(k) = c_1 \sum_{\substack{p_j + q_j = 1 \\[2pt] j = 1,\dots,6}}f(p_j k) g(q_j k) + c_2 \sum_{\substack{p_j + q_j = 1 \\[2pt] j = 7,\dots,12}}f(p_j k) g(q_j k), \quad k \in \mathbb{\Lambda},
		\label{EQ:product_1D_2}
		\end{equation}
		where $c_1$ and $c_2$ are arbitrary real prefactors.
		\label{THE:product_classification}
	\end{theorem}
	\begin{proof}
		Properties~\ref{DEF:PROD_reality}, \ref{DEF:PROD_bilinearity}, \ref{DEF:PROD_leibniz} and symmetries~\ref{DEF:PROD_scaling} and~\ref{DEF:PROD_isotropy} were already used to reduce the product to the form~\eqref{EQ:general_product}.
		One may check that the remaining conditions~\ref{DEF:PROD_commutativity} and~\ref{DEF:PROD_associativity_avg} for the product can be written as linear equations with unit coefficients with respect to the variables $c_j$.
		The system of such equations can be solved explicitly, leading to formulas~\eqref{EQ:product_1D_1} and~\eqref{EQ:product_1D_2}.
		Consider, for example, the case $\lambda = 2$, whose product expression is given by~\eqref{EQ:product_DN_general}.
		Commutativity~\ref{DEF:PROD_commutativity} requires $c_1 = c_2$.
		On the other hand, associativity in average~\ref{DEF:PROD_associativity_avg} enforces all coefficients to be the same. 
	\end{proof}
	
	Recall that the associativity condition is valid in average; see property~\ref{DEF:PROD_associativity_avg} in Definition~\ref{DEF:PROD}. At the same time, the products cannot be associative, as it follows from
	
	\begin{corollary}\label{THE:product_associativity}
		The non-trivial products described in Theorem~\ref{THE:product_classification} are not associative: condition $(f \ast g) \ast h = f \ast (g \ast h)$ is not valid for all functions $f$, $g$ and $h$.
	\end{corollary}
	\begin{proof}
		Let us show that there are $p,q,r \in \mathbb{\Lambda}$ such that $p+q, p+q+r \in \mathbb{\Lambda}$, but $q+r \notin \mathbb{\Lambda}$.
		From the proof of Theorem~\ref{THE:all_triads}, there are integers $0 \leq a < b$ such that the spacing $\lambda$ satisfies $1 = \lambda^b - \lambda^a$.
		Take $p = \lambda^{2b}$, $q = -\lambda^{a+b}$ and $r = -\lambda^{a}$.
		Then $p+q+r = 1 \in \mathbb{\Lambda}$ and $p+q = \lambda^b \in \mathbb{\Lambda}$.
		We claim that $q+r = -(1+\lambda^b)\lambda^a \notin \mathbb{\Lambda}$, which is equivalent to the condition $1+\lambda^b \notin \mathbb{\Lambda}$. 
		Indeed, suppose that $1+\lambda^b \in \mathbb{\Lambda}$.
		In this case, $1+\lambda^b = \lambda^m$ for some integer $m>b$. 
		It follows that $\lambda$ is a common root of trinomials~\eqref{EQ:trinomial} with $(a_1,b_1) = (a,b)$ and $(a_2,b_2) = (b,m)$.
		However, such a solution is forbidden by Lemma~\ref{LEM:solution_trinomials}, leading to a contradiction.
		Now, indicating by $\delta_k$ the function with $\delta_k(k) = 1$ and zero elsewhere, it follows from expression~\eqref{EQ:product_1D_1} that $(\delta_p \ast \delta_q) \ast \delta_r = \delta_{p+q} \ast \delta_r = \delta_{p+q+r}$, but $\delta_p \ast (\delta_q \ast \delta_r) = \delta_p \ast 0 = 0$.
		A similar argument applies to expression~\eqref{EQ:product_1D_2}.
	\end{proof}
	
	Application of the same ideas for the two and three-dimensional cases yield similar formulas for products on these spaces, but with a larger number of free coefficients.
	For instance, the product on the three-dimensional lattice with spacing $\lambda = \varphi$, the golden mean, has $10$ free real coefficients.
	
	It is useful to give the following expression
	\begin{equation}
	(f \ast g)(\mathbf{k}) = \sum_{\substack{\mathbf{p} + \mathbf{q} = \mathbf{k}\\[2pt] \mathbf{p},\mathbf{q} \in \mathbb{\Lambda}^d}} f(\mathbf{p}) g(\mathbf{q}), \quad \mathbf{k} \in \mathbb{\Lambda}^d,
	\label{EQ:Product_nD}
	\end{equation}
	analogous to~\eqref{EQ:product_1D_1}, which yields a product in any dimension and any lattice.
	
	All operations introduced in this section are implemented in \textsc{LogLatt}, an efficient \textsc{Matlab}\textsuperscript{\circledR} library for the numerical calculus on logarithmic lattices~\cite{campolina2020loglattmatlab}.
	
	\section{Generalized lattices and products}
	\label{SEC:generalized}
	
	In this section we discuss some generalizations of logarithmic lattices, which can be useful for applications.
	In order to mimic non-local interactions, one can add the origin to the logarithmic lattice
	\begin{equation}
	\mathbb{\Lambda} = \{ 0 \} \cup \{\pm \lambda^n \}_{n \in \mathbb{Z}}.
	\end{equation}
	In this case, every point $k \in \mathbb{\Lambda}$ interacts with the zero node: $k = k + 0 = 0 + k$, which provides additional (non-local) terms to the products. 
	The value $f(0)$ is interpreted as the mean value of $f$ in physical space, in analogy with the same value for continuous functions $\hat{F}(0) = \int F(x)dx$.
	
	The same relations (\ref{EQ:reality})--(\ref{EQ:L1}) and Definition~\ref{DEF:PROD} are used to define the product and other operations.
	For example, when $\lambda = 2$, the product~\eqref{EQ:product_1D_1} at $k \neq 0$ generalizes to
	\begin{multline}
	(f \ast g)(k) = [f(2 k)g(-k) + f(-k)g(2 k) + f(2^{-1}k)g(2^{-1}k)]
	+ c[f(k)g(0) + f(0)g(k)],
	\label{EQ:product_origin}
	\end{multline}
	with an arbitrary real parameter $c$.
	The product $f \ast g$ evaluated at $k = 0$ is given by
	\begin{equation}
	(f \ast g)(0) = c\sum_{k \in \mathbb{\Lambda}}f(k)g(-k)
	\label{EQ:product_at_zero}
	\end{equation}
	with the same prefactor $c$, which is the consequence of associativity in average -- see~\ref{DEF:PROD_associativity_avg} of Definition~\ref{DEF:PROD}. It is natural to set $c = 1$, in which case expression~\eqref{EQ:product_at_zero} coincides with the inner product~\eqref{DEF:inner}, i.e., $(f \ast g)(0) = (f,g)$.
	
	Furthermore, we can define \textit{generalized logarithmic lattices} as arbitrary subsets $\mathbb{\Lambda}' \subset \mathbb{\Lambda}^d$ of logarithmically distributed nodes.
	To ensure that functions satisfying the reality condition~\eqref{EQ:reality} can be represented in $\mathbb{\Lambda}'$, we impose the property that if $\mathbf{k} \in \mathbb{\Lambda}'$ then $-\mathbf{k} \in \mathbb{\Lambda}'$.
	This is the case, for example, of a truncated lattice with a finite number of points
	\begin{equation}
	\mathbb{\Lambda}' = \{ 0, \pm 1, \pm \lambda, \dots, \pm \lambda^N \},
	\label{EQ:generalized_lattice}
	\end{equation}
	or the same subset excluding zero.
	Since a generalized lattice $\mathbb{\Lambda}'$ is not necessarily scaling invariant or isotropic, we cannot demand the corresponding product to have these symmetries.
	Therefore, a product on $\mathbb{\Lambda}'$ is an operation satisfying properties~\ref{DEF:PROD_reality}--\ref{DEF:PROD_leibniz} of Definition~\ref{DEF:PROD}.
	In the following Theorem, we provide one natural form of the product that serves for all generalized lattices.
	
	\begin{theorem}
		Let $\mathbb{\Lambda}' \subset \mathbb{\Lambda}^d$ be a generalized $d$-dimensional logarithmic lattice.
		Then, operation
		\begin{equation}
		(f \ast g)(\mathbf{k}) = \sum_{\substack{\mathbf{p} + \mathbf{q} = \mathbf{k}\\[2pt] \mathbf{p},\mathbf{q} \in \mathbb{\Lambda}'}} f(\mathbf{p}) g(\mathbf{q}), \quad \mathbf{k} \in \mathbb{\Lambda}',
		\label{EQ:Product}
		\end{equation}
		defines a product on $\mathbb{\Lambda}'$ with properties~\ref{DEF:PROD_reality}--\ref{DEF:PROD_leibniz}.
		\label{THE:generalized_product}
	\end{theorem}
	\begin{proof}
		Properties~\ref{DEF:PROD_reality}--\ref{DEF:PROD_leibniz} are directly verified, except for the associativity in average~\ref{DEF:PROD_associativity_avg}, which follows from the fact that both $(f \ast g, h)$ and $(f,g \ast h)$ can be written in the same form as
		\begin{equation}
		\sum_{\substack{\mathbf{p} + \mathbf{q} + \mathbf{r} = \mathbf{0}\\[2pt]\mathbf{p},\mathbf{q},\mathbf{r} \in \mathbb{\Lambda}'}} f(\mathbf{p}) g(\mathbf{q}) h(\mathbf{r}).
		\end{equation}
	\end{proof}
	
	Note that when $\lambda = 2$ and we let $N \to \infty$, the lattice~\eqref{EQ:generalized_lattice} establishes a decimation of Fourier space for $2\pi$-periodic functions, in the spirit of e.g.~\cite{frisch2012turbulence,buzzicotti2016lagrangian,buzzicotti2016intermittency}.
	Other examples, also for $\lambda = 2$ are~\cite{eggers1991does,grossmann1996developed}.
	The application of lattice operations to the one-dimensional Burgers equation reproduces some well-known shell models of turbulence;
	see Appendix~\hyperref[app:A]{A} for the details.
	
	\section{Ideal incompressible flow}\label{SEC:ideal_flow}
	
	In this and next sections, we make sense of incompressible hydrodynamics on logarithmic lattices by applying the operations introduced previously.
	We will consider a $d$-dimensional logarithmic lattice $\mathbb{\Lambda}^d$, for $d = 2$ or $3$, where
	\begin{equation}
	\mathbb{\Lambda} = \{ 0, \pm 1, \pm \lambda, \pm \lambda^2, \dots \},
	\label{EQ:Euler_lattice}
	\end{equation}
	for some $\lambda$ from Theorem~\ref{THE:all_triads}.
	This lattice mimics Fourier space of a system with largest integral scale $L \sim 2 \pi$ corresponding to $|\mathbf{k}| \sim 1$.
	Our derivations below are equally valid for the case $\mathbb{\Lambda} = \{ \pm 1, \pm \lambda, \pm \lambda^2, \dots \}$, where zero is excluded from~\eqref{EQ:Euler_lattice}.
	
	This section is subdivided as follows.
	Section~\ref{SEC:governing_equations} introduces the incompressible Euler equations on the logarithmic lattice and enumerates their main properties.
	Section~\ref{SEC:local_theory} establishes rigorous results concerning the local-in-time existence and uniqueness of strong solutions and the criterion for singularity formation in this model.
	Section~\ref{SEC:blowup} presents a numerical study of blowup in the three-dimensional equations.
	
	\subsection{Basic equations, symmetries and conservation laws}\label{SEC:governing_equations}
	
	We represent the velocity field $\mathbf{u}(\mathbf{k},t) = (u_1,\dots,u_d) \in \mathbb{C}^d$ as a function of the wave vector $\mathbf{k} \in \mathbb{\Lambda}^d$ and time $t \in \mathbb{R}$.
	Similarly we define the scalar pressure $p(\mathbf{k},t)$.
	The inner product for vector fields will be understood as $(\mathbf{u},\mathbf{v}) = (u_1,v_1)+\cdots+(u_d,v_d)$ with the inner product (\ref{DEF:inner}) for each scalar component.
	All functions are supposed to satisfy the reality condition~\eqref{EQ:reality}.
	
	For the governing equations, we use the exact form of the incompressible Euler equations
	\begin{equation}
	\partial_t \mathbf{u} + \mathbf{u} \ast \nabla \mathbf{u} = -\nabla p, \quad \nabla \cdot \mathbf{u} = 0,
	\label{EQ:Euler}
	\end{equation}
	which are defined upon the logarithmic lattice $\mathbb{\Lambda}^d$, with the conventional notation $(\mathbf{u} \ast \nabla \mathbf{v})_i = \sum_{j = 1}^{d} u_j \ast \partial_j v_i$ for the product $\ast$ from Theorem~\ref{THE:generalized_product}.
	Introducing the vorticity $\pmb{\omega} = \nabla \times \mathbf{u}$ and taking the curl of equations~\eqref{EQ:Euler}, we may write the Euler equations in vorticity formulation
	\begin{equation}
	\partial_t \pmb{\omega} + \mathbf{u} \ast \nabla \pmb{\omega} = \pmb{\omega} \ast \nabla \mathbf{u}.
	\label{EQ:Euler_vorticity}
	\end{equation}
	In the case of vanishing average velocity $\mathbf{u}(\mathbf{0}) = \mathbf{0}$ at $\mathbf{k} = \mathbf{0}$, the velocity field is recovered from the vorticity through the Biot-Savart law
	\begin{equation}
	\mathbf{u}(\mathbf{k}) = \frac{i\mathbf{k} \times \pmb{\omega}(\mathbf{k})}{|\mathbf{k}|^2} \quad \text{for} \ \mathbf{k} \neq \mathbf{0}; \quad \mathbf{u}(\mathbf{0}) = \mathbf{0}.
	\label{EQ:Biot_Savart}
	\end{equation}
	Moreover, if we take the divergence of equation~\eqref{EQ:Euler} and use the incompressibility condition, then the pressure may be obtained from the velocities by solving the Poisson equation
	\begin{equation}
	-\Delta p = \nabla \cdot (\mathbf{u} \ast \nabla \mathbf{u}).
	\label{EQ:pressure_elliptic}
	\end{equation}
	
	The proposed model retains many properties of the continuous Euler equations, which rely only upon the structure of the equations and elementary operations on the logarithmic lattice, as described in the previous sections.
	These include the basic symmetry groups.
	\begin{theorem}[Symmetry groups of the Euler equations on the logarithmic lattice]\label{THE:Euler_symmetries}
		Let $\mathbf{u}(\mathbf{k},t)$, $p(\mathbf{k},t)$  be a solution of the Euler equations~\eqref{EQ:Euler}.
		Then the following transformations also yield solutions:
		\begin{enumerate}[label=\textit{(E.\arabic*)}]
			\item (Time translations)
			$
			\!
			\begin{aligned}[t]
			\mathbf{u}^\tau(\mathbf{k},t) = \mathbf{u}(\mathbf{k},t+\tau),
			%\ p^\tau(\mathbf{k},t) = p(\mathbf{k},t+\tau),
			\end{aligned}
			$
			for any $\tau \in \mathbb{R}$;
			\label{SYM:time}
			\item (Space translations)
			$
			\!
			\begin{aligned}[t]
			\mathbf{u}^{\pmb{\xi}}(\mathbf{k},t) 
			= e^{-i\mathbf{k}\cdot \pmb{\xi}}\mathbf{u}(\mathbf{k},t),
			%\ p^{\pmb{\xi}}(\mathbf{k},t) = e^{-i\mathbf{k}\cdot \pmb{\xi}}p(\mathbf{k},t),		
			\end{aligned}
			$
			for any $\pmb{\xi} \in \mathbb{R}^d$;
			\label{SYM:space}
			\item (Isotropy and parity)
			$
			\mathbf{u}^R(\mathbf{k},t) = R^{-1}\mathbf{u}(R\mathbf{k},t),
			%\ p^{\mathbf{R}}(\mathbf{k},t) = p(\mathbf{Rk},t)
			$
			where $R \in \mathsf{O_h}$ is any element of the group of cube symmetries
			(cf. Definition~\ref{DEF:PROD});
			\label{SYM:isotropy}
			\item (Scale invariance)
			$
			\!
			\begin{aligned}[t]
			\mathbf{u}^{n,h}(\mathbf{k},t) = \lambda^h \mathbf{u}
			\left(\lambda^{n} \mathbf{k},\lambda^{h-n} t \right),
			%\ p^{n,h}(\mathbf{k},t) = \lambda^{2h} p\left(\lambda^{-n} \mathbf{k},\lambda^{h-n} t \right),
			\end{aligned}
			$
			for any $h \in \mathbb{R}$ and $n \in \mathbb{Z}$, 
			where $\lambda$ is the lattice spacing;
			\label{SYM:scaling}
			\item (Time reversibility)
			$
			\!
			\begin{aligned}[t]
			\mathbf{u}^{r}(\mathbf{k},t) = -\mathbf{u}
			\left(\mathbf{k},-t \right)
			%\ p^r(\mathbf{k},t) = p(\mathbf{k},-t)
			\end{aligned}
			$;
			\label{SYM:timerev}
			\item (Galilean invariance)
			$
			\!
			\begin{aligned}[t]
			\mathbf{u}^{\mathbf{v}}(\mathbf{k},t) = e^{-i\mathbf{k}\cdot \mathbf{v}t}
			\mathbf{u}(\mathbf{k},t) - \widehat{\mathbf{v}}(\mathbf{k}),
			%\ p^{\mathbf{v}}(\mathbf{k},t) = e^{-i\mathbf{k}\cdot \mathbf{v}t}p(\mathbf{k},t),
			\end{aligned}
			$
			for any $\mathbf{v} \in \mathbb{R}^d$, where $\widehat{\mathbf{v}}(\mathbf{k})$ is the constant velocity field on the lattice defined as $\widehat{\mathbf{v}}(\mathbf{0}) = \mathbf{v}$ and zero for $\mathbf{k} \neq \mathbf{0}$.
			\label{SYM:galilean}
		\end{enumerate}
		\noindent 
		We did not write the transformations for the pressure $p$ because it can be eliminated from the Euler equations.
		\label{THM:sym_spec}
	\end{theorem}
	
	Recall that the factors $e^{-i\mathbf{k}\cdot \pmb{\xi}}$ and $e^{-i\mathbf{k}\cdot \mathbf{v}t}$ in the symmetries \ref{SYM:space} and \ref{SYM:galilean} are Fourier representations of physical-space translations by the vectors $\pmb{\xi}$ and $\mathbf{v}t$. Thus, the listed symmetries of the Euler equations on the logarithmic lattice are the same as those for the continuous model, except that isotropy~\ref{SYM:isotropy} and scale invariance~\ref{SYM:scaling} are given in discrete form.
	
	Model~\eqref{EQ:Euler} also preserves the same invariants as the continuous Euler equations.
	Let us show this first for the energy and for the enstrophy or helicity, in the two or three-dimensional cases respectively.
	Here we proceed formally.
	The proofs in this section hold for strong solutions, whose existence and uniqueness for short times are established in the next Section~\ref{SEC:local_theory}.
	
	\begin{theorem}[Conservation of energy, enstrophy, helicity]
		Let $\mathbf{u}(t)$ be a solution of the three-dimensional Euler equations~\eqref{EQ:Euler}.
		Then the energy
		\begin{equation}
		E(t) = \frac{1}{2} (\mathbf{u}, \mathbf{u})
		\label{EQ:euler_energy}
		\end{equation}
		and the helicity
		\begin{equation}
		H(t) = (\mathbf{u},\pmb{\omega})
		\label{EQ:helicity}
		\end{equation}
		are conserved in time.
		In the the two-dimensional case, the energy~\eqref{EQ:euler_energy} and the enstrophy
		\begin{equation}
		\Omega(t) = \frac{1}{2}(\pmb{\omega},\pmb{\omega})
		\label{EQ:enstrophy}
		\end{equation}
		are conserved in time.
		\label{THM:conserved_quantities}
	\end{theorem}
	\begin{proof}
		Taking the energy as an example, let us show how the proof can be written using the basic operations defined on the logarithmic lattice, following the standard approach of fluid dynamics.
		Using the Euler equations~\eqref{EQ:Euler}, we obtain
		\begin{equation}
		\frac{dE}{dt} = 
		\displaystyle
		\frac{d}{dt} \left[ \frac{1}{2} (\mathbf{u}, \mathbf{u}) \right] 
		= (\mathbf{u}, \partial_t \mathbf{u}) 
		= -\, (\mathbf{u}, \nabla p) - (\mathbf{u}, \mathbf{u} \ast \nabla \mathbf{u}).
		\end{equation}
		The pressure term vanishes owing to the incompressibility condition as
		\begin{equation}
		(\mathbf{u}, \nabla p)
		= \sum_{i = 1}^{d} (u_i,\partial_i p)
		= -\sum_{i = 1}^{d} (\partial_i u_i, p)
		= - (\nabla \cdot \mathbf{u}, p) = 0,
		\end{equation}
		where the second relation is obtained from the integration by parts~\eqref{EQ:by_parts}.
		In the inertial term, using commutativity of the product~\ref{DEF:PROD_commutativity}, the associativity in average~\ref{DEF:PROD_associativity_avg} and the Leibniz rule~\ref{DEF:PROD_leibniz}, one obtains
		\begin{align}
		(\mathbf{u}, \mathbf{u} \ast \nabla \mathbf{u})
		= \sum_{i,j = 1}^{d}(u_i,  u_j \ast \partial_j u_i)
		= \sum_{i,j = 1}^{d} (u_i \ast \partial_j u_i,  u_j)
		= \frac{1}{2} \sum_{i,j = 1}^{d} (\partial_j (u_i \ast u_i), u_j).
		\end{align}
		After integration by parts, this term vanishes due to the incompressibility condition.
		
		Conservation of enstrophy and helicity in their respective space dimensions can be proved following a similar line of derivations.
	\end{proof}
	
	One can also derive the analogue of Kelvin's Circulation Theorem for the Euler system ~\eqref{EQ:Euler} on a logarithmic lattice.
	For this purpose, let us recall the relation of circulation with the cross-correlation $\Gamma = (\mathbf{u}, \mathbf{h})$ for ``frozen-into-fluid'' divergence-free vector fields $\mathbf{h}(\mathbf{k},t)$ satisfying the equations~\cite{majda2002vorticity}
	\begin{equation}
	\partial_t \mathbf{h} + \mathbf{u} \cdot \nabla \mathbf{h} - \mathbf{h} \cdot \nabla \mathbf{u} = \mathbf{0}, \quad \nabla \cdot \mathbf{h} = 0.
	\label{EQ:frozen_field}
	\end{equation}
	The circulation around a closed material contour $\mathbf{C}(s,t)$ in three-dimensional physical space ($s$ is the arc length parameter) is given by the  cross-correlation $\Gamma$ with the field~\cite{zakharov1997hamiltonian}
	\begin{equation}
	\mathbf{h}(\mathbf{x},t) = \oint \frac{\partial\mathbf{C}(s,t)}{\partial s}\,\delta^3(\mathbf{x}-\mathbf{C}(s,t))\,ds,
	\label{EQ:hLoop}
	\end{equation}
	where $\delta^3$ is the 3D Dirac delta function.
	The field (\ref{EQ:hLoop}) satisfies equations~\eqref{EQ:frozen_field} in the sense of distributions.
	Thus, Kelvin's Theorem follows, as a particular case, from the conservation of cross-correlation $\Gamma$.
	The following Theorem proves the conservation of cross-correlation in the lattice model. 
	
	\begin{theorem}[Kelvin's Theorem]
		Let $\mathbf{u}(t)$ be a solution of the three-dimensional Euler equations~\eqref{EQ:Euler}. Then, for any ``frozen-into-fluid'' divergence-free field $\mathbf{h}(t)$ satisfying equations
		\begin{equation}
		\partial_t \mathbf{h} + \mathbf{u} \ast \nabla \mathbf{h} - \mathbf{h} \ast \nabla \mathbf{u} = \mathbf{0}, \quad \nabla \cdot \mathbf{h} = 0,
		\label{EQ:frozen_field_shell}
		\end{equation}
		the cross-correlation $\Gamma(t) = (\mathbf{u},\mathbf{h})$ is conserved in time.
		\label{THM:Kelvin}
	\end{theorem}
	
	Since equations~\eqref{EQ:frozen_field_shell} are satisfied by the vorticity field $\pmb{\omega}$, the proof for conservation of the cross-correlation follows the same steps as for conservation of helicity~\eqref{EQ:helicity}.
	Theorem~\ref{THM:Kelvin} provides an infinite number of circulation invariants: the cross-correlation $\Gamma$ is conserved for any solution of system~\eqref{EQ:frozen_field_shell}.
	
	For two-dimensional flows, Kelvin's Theorem can be reformulated as the conservation of flux of vorticity across surfaces moving with the fluid.
	This flux can be expressed as the inner product $\Gamma(t) = (a,\omega)$ of the scalar vorticity $\omega = \partial_1 u_2 - \partial_2 u_1$ with a Lagrangian marker $a(\mathbf{k},t)$~\cite{majda2002vorticity}, which is advected by the flow and satisfies the equation $\partial_t a + \mathbf{u} \cdot \nabla a = 0$. Indeed, taking the Lagrangian marker as the indicator function of a bounded surface $S_t$ carried by the flow~\cite[Sec.~1.2]{chorin1990mathematical}, the flux of vorticity across $S_t$ yields the circulation along the contour $\partial S_t$, i.e., $\Gamma(t) = \int_{S_t} \omega dS = \int_{\partial S_t} \mathbf{u} \cdot d\mathbf{l}$.
	On the logarithmic lattice, the vorticity flux is introduced similarly, as the inner product $\Gamma(t) = (a,\omega)$ of the scalar vorticity with a Lagrangian marker satisfying the equation
	\begin{equation}
	\partial_t a + \mathbf{u} \ast \nabla a = 0.
	\label{EQ:marker}
	\end{equation}
	It is straightforward to show that, given the solution $\mathbf{u}(\mathbf{k},t)$ of the two-dimensional Euler system~\eqref{EQ:Euler}, the conservation of $\Gamma$ holds for any solution of (\ref{EQ:marker}).
	
	\subsection{Regularity of solutions}\label{SEC:local_theory}
	
	In this section, we establish the local theory for the Euler system on the logarithmic lattice.
	Here the results are similar to those for the original model: we show local existence and uniqueness of strong solutions and the Beale-Kato-Majda (BKM) blowup criterion~\cite{beale1984remarks}.
	Two-dimensional solutions turn out to be globally regular.
	
	For simplicity, we assume the vanishing average velocity $\mathbf{u}(\mathbf{0}) = \mathbf{0}$ at $\mathbf{k} = \mathbf{0}$ and, therefore, consider only wave vectors with $|\mathbf{k}| \ne 0$ in the following analysis.
	For the lattice variables, we introduce the $\ell^2$ norm in the standard way as $||\mathbf{u}||_{\ell^2} = \left( \sum_{\mathbf{k} \in \mathbb{\Lambda}^d} |\mathbf{u}(\mathbf{k})|^2 \right)^{1/2}$ and the $\ell^\infty$ norm as $||\mathbf{u}||_{\ell^\infty} = \sup_{\mathbf{k} \in \mathbb{\Lambda}^d}|\mathbf{u}(\mathbf{k})|$.
	Given a nonnegative integer $m$, we introduce the operator $D^m$ as
	\begin{equation}
	D^m \mathbf{u}(\mathbf{k}) = |\mathbf{k}|^m \mathbf{u}(\mathbf{k}),
	\end{equation}
	and define the homogeneous Sobolev spaces $h^m$ on the lattice consisting of the functions with finite norm
	\begin{equation}
	||\mathbf{u}||_{h^m} = ||D^m \mathbf{u}||_{\ell^2} = \left( \sum_{\mathbf{k} \in  \mathbb{\Lambda}^d} |\mathbf{k}|^{2m} |\mathbf{u}(\mathbf{k})|^2 \right)^{1/2} <\infty.
	\end{equation}
	Clearly, the space $h^m$ is a Hilbert space endowed with the inner product $(\mathbf{u},\mathbf{v})_{h^m} = (D^m \mathbf{u}, D^m \mathbf{v})$, whose functions have all partial derivatives up to order $m$ in $\ell^2$.
	Finally, we consider the space of divergence-free vector fields
	\begin{equation}
	V^m = \{ \mathbf{u} \in h^m| \nabla \cdot \mathbf{u} = 0 \},
	\end{equation}
	which provides the natural setting for strong solutions of the Euler equations.
	The space $V^m$ is endowed with the $h^m$ norm.
	
	\begin{theorem}\label{THE:local_existence}
		Let $\mathbf{u}^0 \in V^m$ for some $m \geq 1$.
		Then, there exists a time $T>0$, such that the incompressible Euler equations on the logarithmic lattice~\eqref{EQ:Euler} have a unique strong solution $\mathbf{u}(t)$ in the class
		\begin{equation}
		\mathbf{u} \in C^1([0,T);V^m),
		\label{EQ:solution_regularity_class}
		\end{equation}
		with initial condition $\mathbf{u}\big|_{t=0} = \mathbf{u}^0$.
		This solution either exists globally in time, or there is a finite maximal time of existence $t_b$ such that
		\begin{equation}
		\limsup_{t \nearrow t_b} ||\mathbf{u}(t)||_{h^m} = \infty.
		\label{EQ:blowup_Vm_norm}
		\end{equation}
	\end{theorem}
	\begin{proof}
		The proof is similar to that in~\cite{constantin2007regularity} for shell models of turbulence and exploits the locality of the nonlinear interactions on the logarithmic lattice, which turns the convective term into the action of a bounded operator.
		We write the Euler system~\eqref{EQ:Euler} in the functional form
		\begin{equation}
		\partial_t \mathbf{u} + B(\mathbf{u},\mathbf{u}) = -\nabla p,
		\label{EQ:Euler_functional}
		\end{equation}
		where we have introduced the operator
		\begin{equation}
		B(\mathbf{u},\mathbf{v}) = \mathbf{u} \ast \nabla \mathbf{v}.
		\end{equation}
		Operator $B$ is a bounded bilinear operator in $h^m$ -- see the proof in Appendix~\hyperref[app:B]{B}.
		
		Next, in order to eliminate pressure, we project Eq. \eqref{EQ:Euler_functional} onto the space of divergence-free vector fields.
		We introduce the Leray projector $\mathbb{P}$ -- cf. \cite[Sec. 2.1]{robinson2016three} -- on the logarithmic lattice, explicitly given by
		\begin{equation}
		\mathbb{P}_{ij}(\mathbf{k}) = \delta_{ij} - \frac{k_ik_j}{|\mathbf{k}|^2}, \quad \mathbf{k} \in \mathbb{\Lambda}^d.
		\end{equation}
		Since $\mathbf{u}$ is divergence free and $\nabla p$ is a full gradient, it follows that $\mathbb{P}\mathbf{u} = \mathbf{u}$ and $\mathbb{P} \nabla p = 0$, and so we are reduced to the problem
		\begin{equation}
		\frac{d\mathbf{u}}{dt} = F(\mathbf{u}), \quad \mathbf{u}\big|_{t=0} = \mathbf{u}^0,
		\label{EQ:Euler_ODE}
		\end{equation}
		where $F(\mathbf{u}) = -\mathbb{P} B(\mathbf{u},\mathbf{u})$ maps functions from $V^m$ to itself.
		We claim that $F$ is locally-Lipschitz continuous.
		Since $\mathbb{P}$ is an orthogonal projection on $h^m$, and therefore $||\mathbb{P}\mathbf{v}||_{h^m} \leq ||\mathbf{v}||_{h^m}$, we have
		\begin{equation}
		\begin{aligned}
		||F(\mathbf{u}) - F(\mathbf{v})||_{h^m} &= ||\mathbb{P} [B(\mathbf{u},\mathbf{u}) - B(\mathbf{v},\mathbf{v})] ||_{h^m}\\
		&\leq ||B(\mathbf{u},\mathbf{u}) - B(\mathbf{v},\mathbf{v}) ||_{h^m} \\
		&\leq ||B(\mathbf{u},\mathbf{u}-\mathbf{v})||_{h^m} + ||B(\mathbf{u} - \mathbf{v},\mathbf{v}) ||_{h^m}.
		\end{aligned}
		\end{equation}
		In the last inequality, we have applied the bilinearity of $B$ and the triangle inequality.
		Using the boundness of operator $B$, there exists a constant $C>0$ such that
		\begin{equation}
		||B(\mathbf{u},\mathbf{u}-\mathbf{v})||_{h^m} \leq C||\mathbf{u}||_{h^m} ||\mathbf{u} - \mathbf{v}||_{h^m}.
		\end{equation}
		A similar inequality is obtained for the other term $||B(\mathbf{u} - \mathbf{v}, \mathbf{v})||_{h^m}$, which proves the Lipschitz continuity of $F$ when $||\mathbf{u}||_{h^m}$ and $||\mathbf{v}||_{h^m}$ are bounded by some constant.
		
		It follows that Eq.~\eqref{EQ:Euler_ODE} is an ordinary differential equation with $F$ locally-Lipschitz continuous on the Banach space $V^m$.
		In this framework, we apply the Picard Theorem on Banach spaces -- see e.g.~\cite{cartan1967calcul,schechter2004introduction} -- to guarantee existence of a unique local solution in the class~\eqref{EQ:solution_regularity_class} and initial condition $\mathbf{u}^0$.
		The pressure is recovered by solving the Poisson equation~\eqref{EQ:pressure_elliptic}.
		The blowup statement in~\eqref{EQ:blowup_Vm_norm} also follows from classical theory of ordinary differential equations~\cite{schechter2004introduction}.
	\end{proof}
	
	\begin{theorem}[BKM blowup criterion]\label{THE:BKM}
		Let $\mathbf{u}(t) \in C^1([0,t_b);V^m)$ be a strong solution for the incompressible Euler equations~\eqref{EQ:Euler} on the logarithmic lattice, where $t_b$ is the maximal time of existence.
		Then either $t_b = \infty$ or 
		\begin{equation}
		\int_{0}^{t_b}||\pmb{\omega}(t)||_{\ell^\infty} dt = \infty.
		\label{EQ:BKM}
		\end{equation}
		In the later case, we have necessarily
		\begin{equation}
		\limsup_{t \nearrow t_b} ||\pmb{\omega}(t)||_{\ell^\infty} = \infty.
		\label{EQ:BKM_limsup}
		\end{equation}
	\end{theorem}
	\begin{proof}
		Let us assume that
		\begin{equation}
		\int_{0}^{t_b}||\pmb{\omega}(t)||_{\ell^\infty} dt = M < \infty.
		\end{equation}
		for a finite $t_b < \infty$.
		We are going to prove that this implies
		\begin{equation}
		\label{EQ:no49}
		||\mathbf{u}(t)||_{h^m} \leq N, \quad \forall t<t_b,
		\end{equation}
		for some constant $N<\infty$, thus contradicting condition~\eqref{EQ:blowup_Vm_norm} of Theorem~\ref{THE:local_existence}.
		To show this, we perform an energy estimate for Eq.~\eqref{EQ:Euler}.
		We set $\mathbf{v} = D^m \mathbf{u}$ and $q = D^m p$ and apply $D^m$ to Eq.~\eqref{EQ:Euler} to obtain
		\begin{equation}
		\partial_t \mathbf{v} = -D^m(\mathbf{u} \ast \nabla \mathbf{u}) - \nabla q.
		\label{EQ:energy_estimate_step1}
		\end{equation}
		Taking the $\ell^2$-inner product of Eq.~\eqref{EQ:energy_estimate_step1} with $\mathbf{v}$ yields
		\begin{equation}
		\frac{1}{2} \frac{d}{dt}||\mathbf{v}||_{\ell^2}^2 = -(D^m(\mathbf{u} \ast \nabla \mathbf{u}),\mathbf{v}) - (\nabla q, \mathbf{v}).
		\label{EQ:energy_estimate_step2}
		\end{equation}
		After integrating by parts, the last term vanishes due to incompressibility as
		\begin{equation}
		(\nabla q, \mathbf{v}) = \sum_{i = 1}^{d}(\partial_i q, v_i) = -\sum_{i = 1}^{d}(q,\partial_i v_i) = -\sum_{i = 1}^{d}(q,D^m \partial_i u_i) = -(q,D^m \nabla \cdot \mathbf{u}) = 0.
		\label{EQ:no52}
		\end{equation}
		Next, we use the following calculus inequality on logarithmic lattices
		\begin{equation}
		||\mathbf{f} \ast \mathbf{g}||_{h^m} \leq C(||\mathbf{f}||_{h^m} ||\mathbf{g}||_{\ell^\infty} + ||D\mathbf{f}||_{\ell^\infty} ||\mathbf{g}||_{h^{m-1}}), \quad \text{for } \mathbf{f} \in h^m, \mathbf{g} \in h^{m-1}
		\label{EQ:calculus_inequality}
		\end{equation}
		for some positive constant $C$; this inequality has a continuous analogue for Sobolev spaces $H^s$ -- see e.g.~\cite[Sec. 3.2.1]{majda2002vorticity} -- and the lattice version~\eqref{EQ:calculus_inequality} is proved in the Appendix~\hyperref[app:B]{B}.
		Then, the nonlinear term in (\ref{EQ:energy_estimate_step2}) can be estimated using $\mathbf{f} = \mathbf{u}$ and $\mathbf{g} = \nabla u_i$ as
		\begin{equation}
		\begin{array}{rcl}
		(D^m(\mathbf{u} \ast \nabla \mathbf{u}),\mathbf{v}) & \leq & ||\mathbf{u} \ast \nabla \mathbf{u}||_{h^m} ||\mathbf{v}||_{\ell^2} 
		\\[3pt]
		& \leq & \displaystyle
		C||\mathbf{v}||_{\ell^2}\sum_{i = 1}^d\left( 
		||\mathbf{u}||_{h^m} ||\nabla u_i||_{\ell^\infty} 
		+ ||D\mathbf{u}||_{\ell^\infty} ||\nabla u_i||_{h^{m-1}}
		\right)
		\\[15pt]
		& \leq & \displaystyle
		2dC||\mathbf{v}||_{\ell^2} ||\mathbf{u}||_{h^m} ||D\mathbf{u}||_{\ell^\infty}
		=
		C'||\mathbf{v}||_{\ell^2}^2 ||D\mathbf{u}||_{\ell^\infty},
		\end{array}
		\label{EQ:no54}
		\end{equation}
		where at the end we used $||\mathbf{u}||_{h^m} =  ||D^m \mathbf{u}||_{\ell^2} = ||\mathbf{v}||_{\ell^2}$ and set $C' = 2dC$.
		Substituting relations (\ref{EQ:no52}) and (\ref{EQ:no54}) into \eqref{EQ:energy_estimate_step2} yields
		\begin{equation}
		\frac{d}{dt}||\mathbf{v}||_{\ell^2}^2 \leq 2C'||\mathbf{v}||_{\ell^2}^2 ||D\mathbf{u}||_{\ell^\infty},
		\end{equation}
		and applying Gronwall's Inequality, we are lead to
		\begin{equation}
		||\mathbf{v}(t)||_{\ell^2} \leq ||\mathbf{v}(0)||_{\ell^2} \exp \left( C' \int_{0}^{t} ||D\mathbf{u}(s)||_{\ell^\infty}ds \right).
		\end{equation}
		Finally, using the estimate
		\begin{equation}
		||D\mathbf{u}||_{\ell^\infty} \leq ||\pmb{\omega}||_{\ell^\infty},
		\end{equation}
		which follows from the Biot-Savart law~\eqref{EQ:Biot_Savart}, and recalling again that $||\mathbf{v}||_{\ell^2} = ||\mathbf{u}||_{h^m}$, we obtain
		\begin{equation}
		||\mathbf{u}(t)||_{h^m} \leq ||\mathbf{u}(0)||_{h^m} \exp \left( C' \int_{0}^{t_b} ||\pmb{\omega}(s)||_{\ell^\infty}ds \right) \leq N, \quad \forall t \in [0,t_b)
		\end{equation}
		for $N = ||\mathbf{u}(0)||_{h^m} \exp(C'M) < \infty$. This is the inequality (\ref{EQ:no49}), which led us to contradiction.
	\end{proof}
	
	\begin{corollary}
		Strong solutions $\mathbf{u}(t)$ of the two-dimensional incompressible Euler equations~\eqref{EQ:Euler} exist globally in time.
	\end{corollary}
	\begin{proof}
		From Theorem~\ref{THM:conserved_quantities}, strong solutions of the two-dimensional Euler equations conserve the $\ell^2$ norm $||\pmb{\omega}||_{\ell^2}$ of the vorticity.
		Hence, the inequality $||\pmb{\omega}||_{\ell^\infty} \leq ||\pmb{\omega}||_{\ell^2}$ on the lattice prevents condition~\eqref{EQ:BKM_limsup} to take place.
	\end{proof}
	
	\subsection{Blowup in incompressible 3D Euler equations}\label{SEC:blowup}
	
	Whether three-dimensional incompressible Euler flow develops a singularity in finite time (also called \textit{blowup}) remains a challenging open mathematical problem. According to the BKM criterion, the singularity implies a spontaneous generation of infinitely large vorticity.
	Such singularity is anticipated by Kolmogorov's theory of turbulence~\cite{frisch1999turbulence}, which predicts that the vorticity diverges at small scales as $\delta \omega \sim \ell^{-2/3}$ when energy is transferred from integral to viscous scales.
	In this context, blowup could reveal an efficient mechanism for the energy cascade and, for this reason, it is often considered a cornerstone for the theory of turbulence.
	
	In addition to purely mathematical approaches, see e.g.~\cite{chae2008incompressible,tao2016finite} and very recent achievements~\cite{chen2019finite,elgindi2019finite}, the blowup problem was intensively investigated through Direct Numerical Simulations (DNS)~\cite{gibbon2008three,grafke2008numerical,hou2009blow}.
	However, numerical results appear to be rather inconclusive, with the controversy~\cite{kerr1993evidence,hou2006dynamic} only growing with the increase of resolution.
	Naturally, several simplified models have been investigated for understanding possible blowup scenarios, e.g.~\cite{constantin1985simple,uhlig1997singularities,dombre1998intermittency,mailybaev2012renormalization}.
	Despite being rather successful in the study of turbulence~\cite{biferale2003shell} and serving as a useful testing ground for mathematical analysis, e.g.~\cite{katz2005finite,cheskidov2008blow}, these models fall short of reproducing basic features of Euler's blowup phenomenon:
	they lack important properties of Euler's flow, such as incompressibility and conservation of circulation, and often show dynamical behavior atypical for Euler solutions, such as self-similarity~\cite{chae2007nonexistence,chae2013formation}.
	Note that we do not discuss here boundary effects~\cite{luo2013potentially}, which set a different open problem.
	
	Unlike many previous simplified models, the Euler equations on logarithmic lattices retain most structural properties of the original equations, as we showed in Section~\ref{SEC:governing_equations}.
	In the work~\cite{campolina2018chaotic,campolina2019fluid}, we presented a numerical evidence of chaotic blowup in the three-dimensional Euler system on a golden-mean logarithmic lattice. 
	Now we extend these previously reported results by testing the robustness of our conclusions on different lattices.
	For the comparison, we consider the golden mean $\lambda = \varphi$ and the plastic number $\lambda = \sigma$, which provide two lattices $\mathbb{\Lambda}^3$ with increasing resolution -- see Fig.~\ref{FIG:2D_lattice}; here, $\mathbb{\Lambda} = \{ \pm1, \pm \lambda, \pm \lambda^2, \dots \}$ is taken, with no zero component.
	We remark that the spacing factor $\lambda = 2$ does not provide a reliable model for the blowup study, because the incompressibility condition together with a small number of triad interactions cause degeneracies in coupling of different modes.
	
	\textsc{Numerical model.}
	Aiming for the study of blowup, initial conditions are chosen to have nonzero components limited to large scales, with wavenumbers $1 \le |k_i| \le \varphi^2 = (3+\sqrt{5})/2$.
	This corresponds to a box of three excited modes in each direction for the golden mean and four modes for the plastic number lattice spacing. 
	The velocities at these modes are explicitly given in the form
	\begin{equation}
	u_j(\mathbf{k}) = \sum_{m,n = 1}^{3}\frac{|\epsilon_{jmn}|}{2}k_m k_n e^{i\theta_j(\mathbf{k})-|\mathbf{k}|}, \quad \text{for} \quad j=1,2.
	\end{equation}
	Here $\epsilon_{jmn}$ is the Levi-Civita permutation symbol and the phases $\theta_j$ are given by
	\begin{equation}
	\theta_j(\mathbf{k}) = \mathrm{sgn}(k_1)\alpha_j + \mathrm{sgn}(k_2)\beta_j
	+\mathrm{sgn}(k_3)\delta_j  + \mathrm{sgn}(k_1k_2k_3)\gamma_j,
	\end{equation}
	with the constants $(\alpha_1,\beta_1,\delta_1,\gamma_1) = (1,-7,13,-3)/4$ and $(\alpha_2,\beta_2,\delta_2,\gamma_2) = (-1,-3,11,7)/4$.
	The third component of velocity is uniquely defined by the incompressibility condition.
	Clearly, because the nodes of different lattices do not match, it is impossible to test the same initial condition on different lattices.
	
	\begin{figure*}[t]
		\centering
		\includegraphics[width=.79\textwidth]{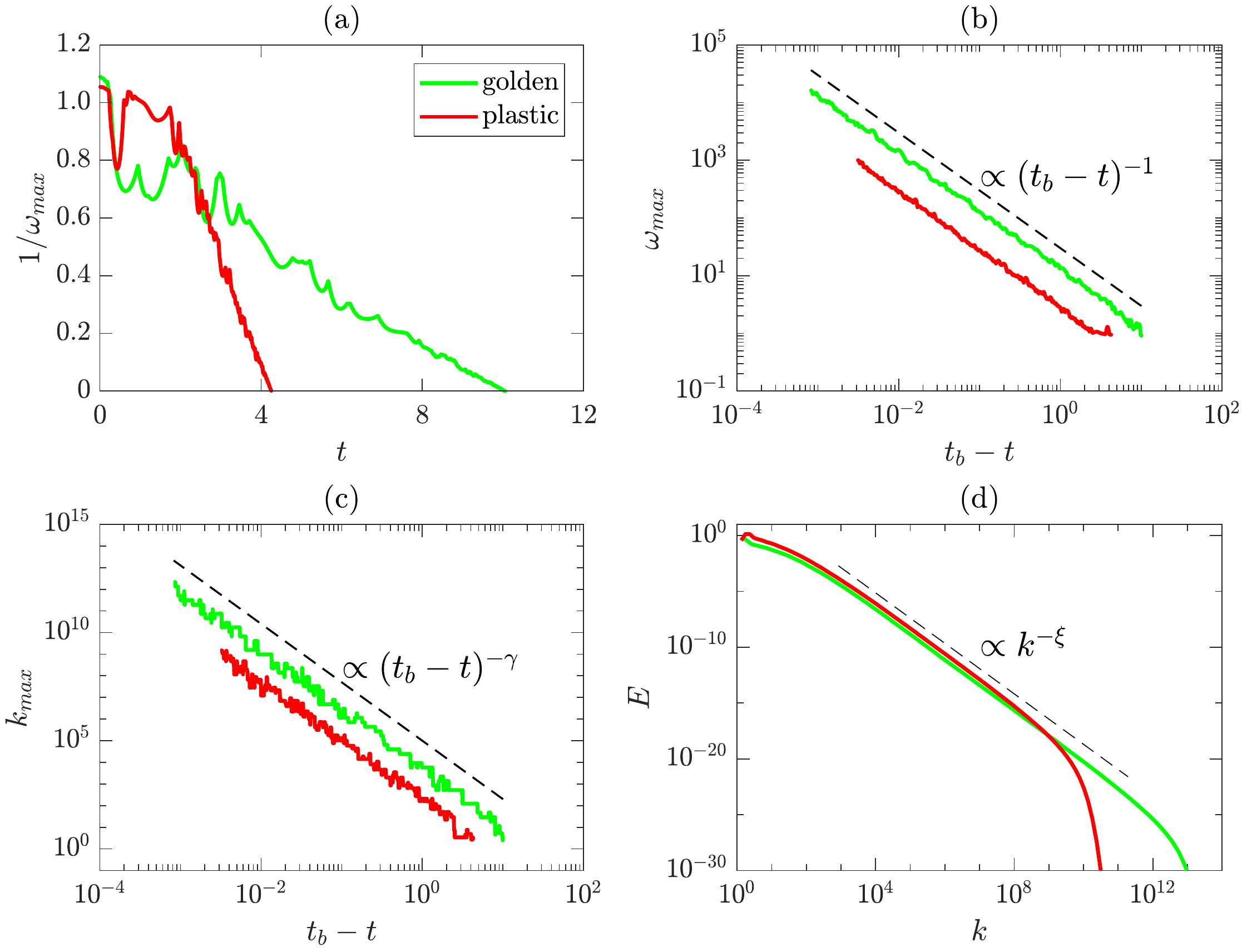}
		\caption{Comparison of the Euler blowup dynamics for different lattice spacings -- the golden mean $\varphi$ in green and plastic number $\sigma$ in red. Owing to the lower computational cost, the simulation for the golden mean spans a larger spatial range. (a) Dynamic evolution of the inverse maximum vorticity $1/\omega_{\max}$, reaching the blowup times $t_b = 4.255$ and $10.052$ for golden and plastic lattice spacings, respectively. (b) The maximum vorticity $\omega_{\max}$ in logarithmic scale fitting in average the power law $\sim (t_b - t)^{-1}$. (c) Wave number $k_{\max}$ where the maximum vorticity occurs in logarithmic scale, following the asymptotic $\sim (t_b -t)^{-\gamma}$ with $\gamma = 2.70$. (d) The energy spectrum~\eqref{EQ:energy_spectrum} at the final time of integration for each simulation, developing the power-law $E(k) \propto k^{-\xi}$, where $\xi = 2.26$.}
		\label{FIG:comp_scalings}
	\end{figure*}
	
	Like in usual DNS, we consider the Euler equations in vorticity formulation~\eqref{EQ:Euler_vorticity}, where the velocity field is recovered through the Biot-Savart law~\eqref{EQ:Biot_Savart}.
	The equations are integrated numerically with double-precision using the fourth-order Runge-Kutta-Fehlberg adaptive scheme~\cite{fehlberg1969low}.
	The time step was dynamically defined in order to keep the relative error for $\omega_{\max}(t) = \max_{\mathbf{k}}|\pmb{\omega}(\mathbf{k},t)|$ below $10^{-6}$.
	Since only a finite number $N$ of modes in each direction can be simulated, the infinite-dimensional nature of the problem was tracked by implementing the following spatial adaptive scheme.
	At each time step, we compute the enstrophy $\Omega(t) = \frac{1}{2} (\pmb{\omega},\pmb{\omega})$ due to the modes with wave vectors $|\mathbf{k}| \geq K_{\max}/\lambda$, where $K_{\max}$ is the largest wave number in each direction.
	This quantity estimates the enstrophy error (i.e., $\ell^2$ norm of vorticity) due to the mode truncation and it was kept below $10^{-15}$ during the whole simulation.
	Every time this threshold was reached we increased the number of modes in each direction by $5$, i.e., multiplying $K_{\max}$ by $\lambda^5$.
	We stopped the simulation for the plastic number with $N = 95$, thus covering a spatial range of $K_{\max} = \sigma^{95} \approx 10^{11}$.
	Due to the higher spacing value, the golden mean allows to cover a larger spatial range with less modes.
	In this case, the simulation was stopped with $N = 70$, which corresponds to $K_{\max} = \varphi^{70} \approx 10^{14}$.
	For the simulations of both lattice spacings, the energy was conserved during the whole time of integration with a relative error below $10^{-6}$.
	
	\textsc{Results.}
	Before presenting our new results, we briefly review the previous conclusions reported in~\cite{campolina2018chaotic}, where a simulation of ideal flow was conducted on a golden mean logarithmic lattice.
	A large amplification of maximum vorticity $\omega_{\max}$ within a finite time $t_b$ was observed, demonstrating an asymptotic blowup solution $\omega_{\max}(t) \sim (t_b - t)^{-1}$, followed by a power-law development in the energy spectrum as $E(k) \propto k^{-\xi}$, $\xi \approx 2.26$.
	Such blowup is linked to a chaotic wave in a renormalized system, traveling with constant average speed $\gamma \approx 2.70$.
	The chaotic behavior justifies the high sensitivity to perturbations, which is encountered in full DNS~\cite{grafke2008numerical} and has theoretical foundation in developed turbulence~\cite{ruelle1979microscopic}.
	Instability of blowup solutions is also observed in other simplified models~\cite{uhlig1997singularities,mailybaev2012c,de2017chaotic} and was proved recently for the full incompressible 3D Euler equations~\cite{vasseur2019blow}.
	The chaotic attractor restores the isotropy in the statistical sense, even though the solution at each particular moment is essentially anisotropic, in similarity to the recovery of isotropy in the Navier-Stokes turbulence~\cite{frisch1999turbulence,biferale2005anisotropy}.
	The striking property of the attractor is its multi-scale character, covering six decades in Fourier space, which seems impossible to be reproduced by the modern numerical techniques in full DNS.
	At the respective scales, the solution of the logarithmic model displays properties that can be associated with typical coherent structures of full DNS, e.g. the effect of two-dimensional depletion~\cite{pumir1990collapsing,agafontsev2017asymptotic}.
	For more details, see~\cite{campolina2018chaotic}.
	
	Fig.~\ref{FIG:comp_scalings} compares the numerical integrations of the Euler equations for golden and plastic lattice spacings.
	Though solutions are, of course, different at earlier times, they demonstrate very close (numerically indistinguishable) asymptotic blowup dynamics, which we analyze in details now.
	Figs.~\ref{FIG:comp_scalings}(a) and~\ref{FIG:comp_scalings}(b) show the dynamic evolution of $\omega_{\max}(t)$.
	BKM blowup criterion -- see Theorem~\ref{THE:BKM} -- requires an asymptotic growth of at least $\omega_{\max}(t) \sim (t_b - t)^{-1}$ as $t$ approaches the blowup time $t_b$.
	This is verified for both simulations by plotting the inverse value $1/\omega_{\max}(t)$ in Fig.~\ref{FIG:comp_scalings}(a), providing the blowup times $t_b = 4.255 \pm 0.001$ and $t_b = 10.052 \pm 0.001$ for the plastic number and golden mean, respectively.
	Fig.~\ref{FIG:comp_scalings}(b) shows the same results in logarithmic scale verifying the asymptotic $\omega_{\max}(t) \sim (t_b - t)^{-1}$.
	Note that a growth of five orders of magnitude is observed for the golden mean.
	The wave number $k_{\max}$ at which the maximum vorticity occurs also grows asymptotically as $k_{\max} \sim (t_b - t)^{-\gamma}$ with the same exponent $\gamma = 2.70 \pm 0.01$ for the two simulations, as shown in Fig.~\ref{FIG:comp_scalings}(c).
	At last, the energy spectrum
	\begin{equation}
	E(k) = \frac{1}{2 \Delta} \sum_{k \leq |\mathbf{k}'| < \lambda k} |\mathbf{u}(\mathbf{k}')|^{2}, \quad \text{with } \Delta = \lambda k - k,
	\label{EQ:energy_spectrum}
	\end{equation}
	develops a power-law $E(k) \propto k^{-\xi}$.
	A dimensional argument~\cite{campolina2018chaotic} predicts the exponent $\xi$ depending upon $\gamma$ as $\xi = 3 - 2/\gamma \approx 2.26$, confirmed in Fig.~\ref{FIG:comp_scalings}(d).
	
	\textsc{Chaotic attractor.}
	We argued in~\cite{campolina2018chaotic} that the blowup dynamics in the Euler system on a logarithmic lattice is associated to a chaotic wave in a renormalized system.
	This chaotic behavior exemplifies the fundamental instability, which is necessary for blowup in the incompressible 3D Euler equations, as proved recently in~\cite{vasseur2019blow}.
	Here, we compare the attractors for the two simulations. These attractors are visualized using the renormalized variables
	\begin{equation}
	\widetilde{\pmb{\omega}} = (t_b-t)\pmb{\omega}, \quad \eta = \log|\mathbf{k}|, \quad
	\mathbf{o} = \mathbf{k}/|\mathbf{k}|, \quad \tau = -\log(t_b-t),
	\label{Renorm}
	\end{equation}
	which apply similarly in Fourier space $\mathbb{R}^3$ and in our lattice $\mathbb{\Lambda}^3$.
	The Euler equations can be rewritten as a dynamical system with respect to these new coordinates; consult~\cite{campolina2018chaotic} for more details. 
	Fig.~\ref{FIG:attractor} shows the time evolution, in renormalized variables, of the solutions on the two different lattices.
	For the comparison, we plot the vorticities $\widetilde{\pmb{\omega}}$ normalized with respect to their correspondent maximum values $\widetilde{\omega}_{\max}$.
	The renormalized time for the plastic number is shifted $\tau \mapsto \tau + \tau_0$ by $\tau_0 = -1.2$ for the attractors to be aligned in space.
	
	\begin{figure*}[t]
		\centering
		\includegraphics[width=\textwidth]{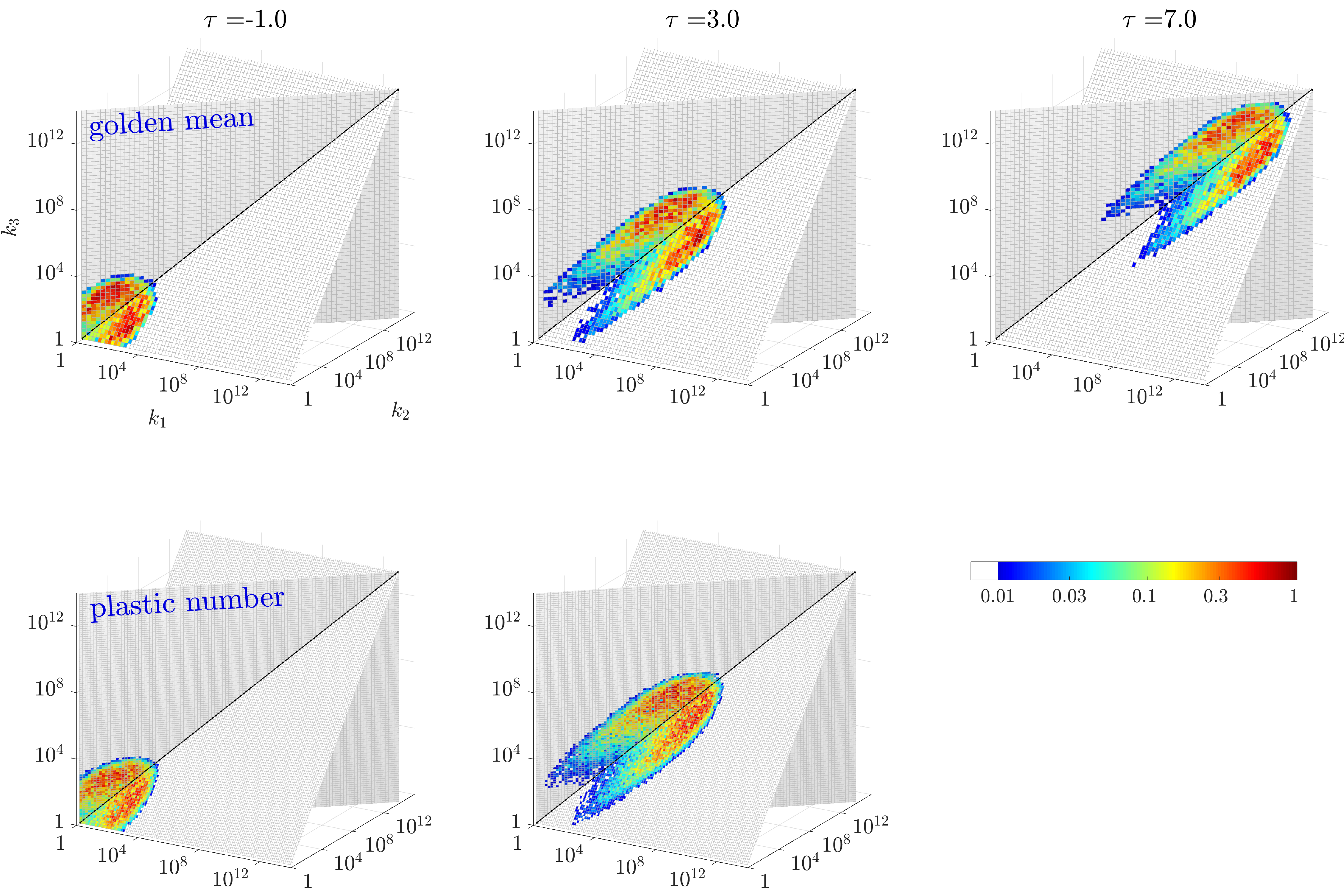}
		\caption{Absolute value of renormalized vorticities $|\widetilde{\pmb{\omega}}|$ plotted on sections of 3D Fourier space, in logarithmic scales, at three different instants $\tau$. For comparison, the vorticities are normalized with respect to their maximum norm $\widetilde{\omega}_{\max}$; values below $0.01$ are plotted in white. The first row shows the evolution on the golden and the second row on the plastic lattice. Owing to the lower computational cost, the simulation for the golden mean was integrated for longer renormalized times $\tau$.}
		\label{FIG:attractor}
	\end{figure*}
	
	The solutions in Fig.~\ref{FIG:attractor} show convergence to chaotic waves, which travel through the main diagonal of Fourier space with the same constant speed $\gamma$.
	In~\cite{campolina2018chaotic}, we demonstrated their chaotic nature by computing a positive largest Lyapunov exponent.
	The chaotic attractors look surprisingly similar despite the quite distinct resolutions furnished by the two lattices.
	
	\section{Viscous incompressible flow and turbulence}
	\label{SEC:viscous_flow}
	
	In this section, we introduce a viscous dissipative term and a forcing $\mathbf{f}$ into the Euler equations~\eqref{EQ:Euler}, leading to the incompressible 3D Navier-Stokes equations on a logarithmic lattice
	\begin{equation}
	\partial_t \mathbf{u} + \mathbf{u} \ast \nabla \mathbf{u} = -\nabla p + \nu \Delta \mathbf{u} + \mathbf{f}, \quad \nabla \cdot \mathbf{u} = 0,
	\label{EQ:Navier-Stokes}
	\end{equation}
	where $\nu \geq 0$ is the kinematic viscosity.
	We will focus on testing some fundamental properties of hydrodynamic turbulence, when the viscous term is responsible for dissipating energy at small scales of the flow while the force injects it at large scales.
	Following the same lines of derivations as for the continuous model, we deduce the balance for the energy~\eqref{EQ:euler_energy} as
	\begin{equation}
	\frac{dE}{dt} = -2\nu \Omega(t) + F(t),
	\end{equation}
	where $\Omega(t)$ is the enstrophy~\eqref{EQ:enstrophy} and $F(t) = (\mathbf{u},\mathbf{f})$ is the work done by external forces.
	The term $\varepsilon = 2\nu \Omega$ is the total dissipation rate of the flow.
	
	\subsection{Anomalous dissipation}\label{SEC:anomalous}
	
	A major feature of turbulent flows is the non-vanishing energy dissipation rate $\epsilon >0$ in the limit of large Reynolds numbers, which can also be formulated mathematically as the limit of vanishing viscosity $\nu \to 0$.
	This apparently paradoxical phenomenon is known as \textit{dissipation anomaly}~\cite{onsager1949statistical,eyink2006onsager} and has found confirmation in many experiments~\cite{pearson2002measurements} and numerical simulations~\cite{kaneda2003energy}.
	
	Dissipation anomaly is conveniently quantified by considering the evolution of energy through different scales. 
	We derive from Eq.~\eqref{EQ:Navier-Stokes} the scale-by-scale energy budget equation
	\begin{equation}
	\partial_t E_k = \Pi_k -2\nu \Omega_k + F_k.
	\end{equation}
	Here, using the notation
	\begin{equation}
	(f,g)_k = \sum_{|\mathbf{k}'| \leq k} f(\mathbf{k}')\overline{g(\mathbf{k}')},
	\end{equation}
	we have introduced the \textit{cumulative energy} between wave number $0$ and $k$
	\begin{equation}
	E_k = \frac{1}{2} (\mathbf{u},\mathbf{u})_k,
	\end{equation}
	the \textit{cumulative enstrophy}
	\begin{equation}
	\Omega_k = \frac{1}{2} (\pmb{\omega},\pmb{\omega})_k,
	\end{equation}
	the \textit{cumulative energy injection}
	\begin{equation}
	F_k = (\mathbf{u}, \mathbf{f})_k,
	\end{equation}
	and the \textit{energy flux}
	\begin{equation}
	\Pi_k = -(\mathbf{u}, \mathbf{u} \ast \nabla \mathbf{u} + \nabla p)_k.
	\label{EQ:energy_flux}
	\end{equation}
	Statistical steady state in a turbulent flow is achieved when $\partial_t \langle E_k \rangle = 0$.
	In this regime, the mean energy flux $\langle \Pi_k \rangle$ balances with the mean energy dissipation $\langle -2\nu \Omega_k \rangle$ and the work of external forces $\langle F_k \rangle$.
	Since for small viscosities it is typical to have energy injection confined to large scales and energy dissipation confined to small scales, a dissipation anomaly is related to the development of a constant energy flux in the intermediate range called the \textit{inertial interval}.
	In our definition, a positive energy flux corresponds to a (direct) cascade of energy from large to small scales.
	
	In order to compute the energy flux, we consider the Navier-Stokes equations~\eqref{EQ:Navier-Stokes} on the three-dimensional logarithmic lattice of spacing $\lambda = \varphi$, the golden mean.
	The energy is injected at large scales $\varphi \leq |k_{1,2,3}| \leq \varphi^3$ through a constant-in-time force with randomly generated components.
	To obtain an extended inertial interval, the viscous forces $\nu \Delta \mathbf{u}$ were replaced by a hyper-viscous term $-\nu (-\Delta)^{h}\mathbf{u}$ with $h = 2$.
	For models with local triad interactions, it is expected that the dynamical statistics are ultraviolet robust, i.e., does not depend on the detailed dissipation mechanism at small scales~\cite{benzi1999intermittency,lvov1998universal}.
	The model was integrated with double-precision using the first-order exponential time-splitting method~\cite{cox2002exponential}: at each time step, we first use the fourth-order Runge-Kutta method to integrate the Euler equations and next we multiply the resulting solution by the exponential factor $e^{(-\nu |\mathbf{k}|^{2h}\Delta t)}$, where $\Delta t$ is the time step.
	
	Fig.~\ref{FIG:flux} shows the mean energy flux $\langle \Pi_k \rangle$ along scales $k$ for different viscosities.
	The energy flux reaches the same constant positive value for all viscosities and the inertial range extends to smaller scales as the viscosity decreases, which indicates the development of a dissipation anomaly in the limit $\nu \to 0$.
	
	\begin{figure*}[t]
		\centering
		\includegraphics[width=.6\textwidth]{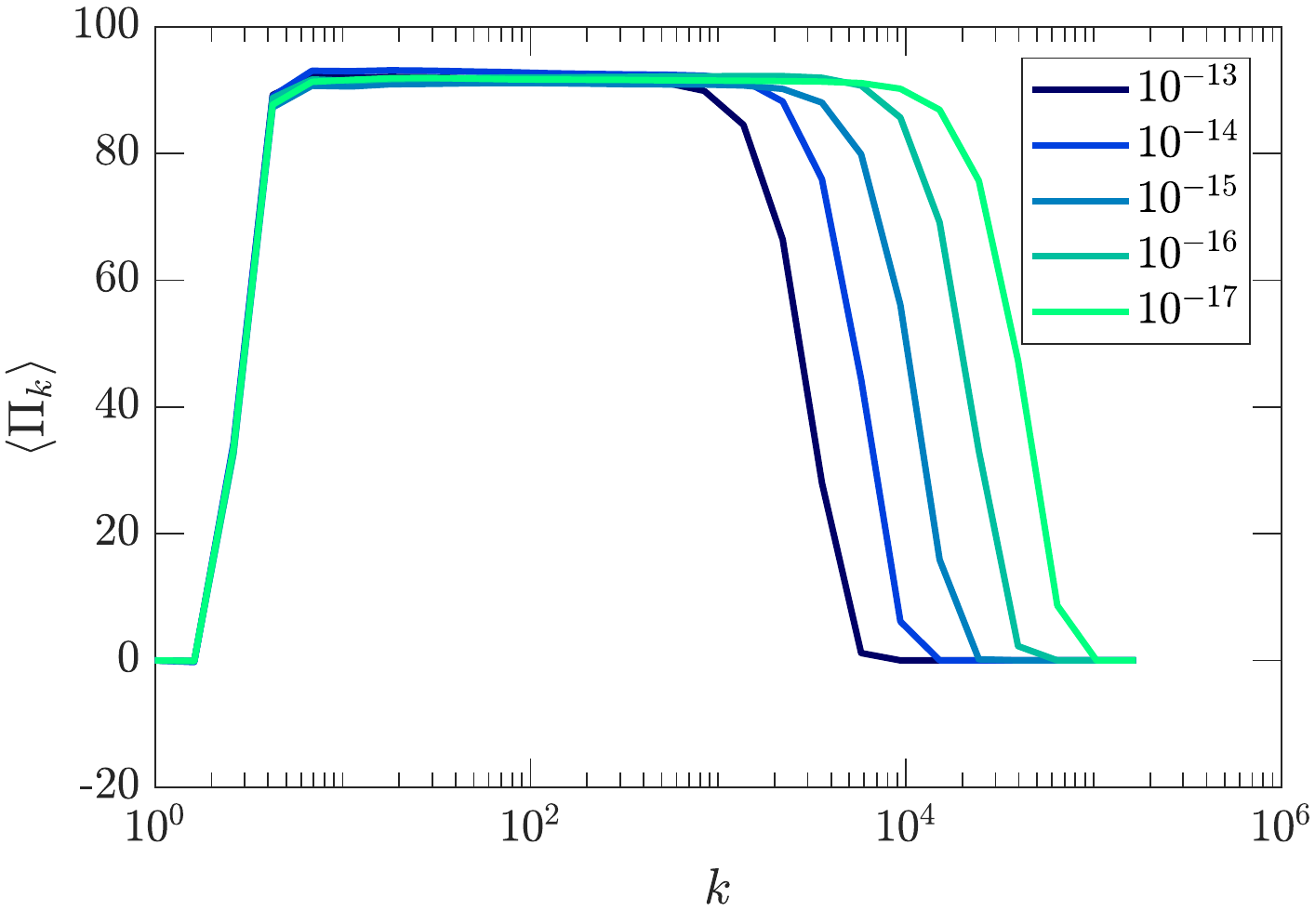}
		\caption{Mean energy flux $\langle \Pi_k \rangle$ along wave numbers $k$, in logarithmic scale, for different values of the hyper-viscous parameter $\nu = 10^{-13}$, $10^{-14}$, $10^{-15}$, $10^{-16}$ and $10^{-17}$. In our notation, a positive energy flux corresponds to a direct cascade of energy.}
		\label{FIG:flux}
	\end{figure*}
	
	\subsection{Statistics of Fourier modes}
	
	The Navier-Stokes equations on the golden mean lattice were integrated using the same numerical procedure described in Section~\ref{SEC:anomalous}, with hyper-viscous term and viscosity $\nu = 10^{-13}$.
	The probability distribution functions (PDF's) were numerically estimated through a histogram binning procedure using the statistics accumulated within a sample time $T$.
	In terms of turnover time $T_0 = 1/|\mathbf{k}_0|U_0$, where $\mathbf{k}_0 = (1,1,1)$ is the wave vector of integral scale and $U_0 = \langle |\mathbf{u}(\mathbf{k}_0)|^2 \rangle^{1/2}$, the sample time $T$ was larger than $90 T_0$.
	The PDF's of $\text{Re}[u_1(\mathbf{k}_n)]$, in units of their root-mean-square $\langle \text{Re}[u_1(\mathbf{k}_n)]^2 \rangle^{1/2}$ values are shown in Fig.~\ref{FIG:PDF}, for several wave vectors $\mathbf{k}_n = \lambda^n \mathbf{k}_0$ rescaled along the main diagonal of Fourier space.
	
	\begin{figure*}[t]
		\centering
		\includegraphics[width=\textwidth]{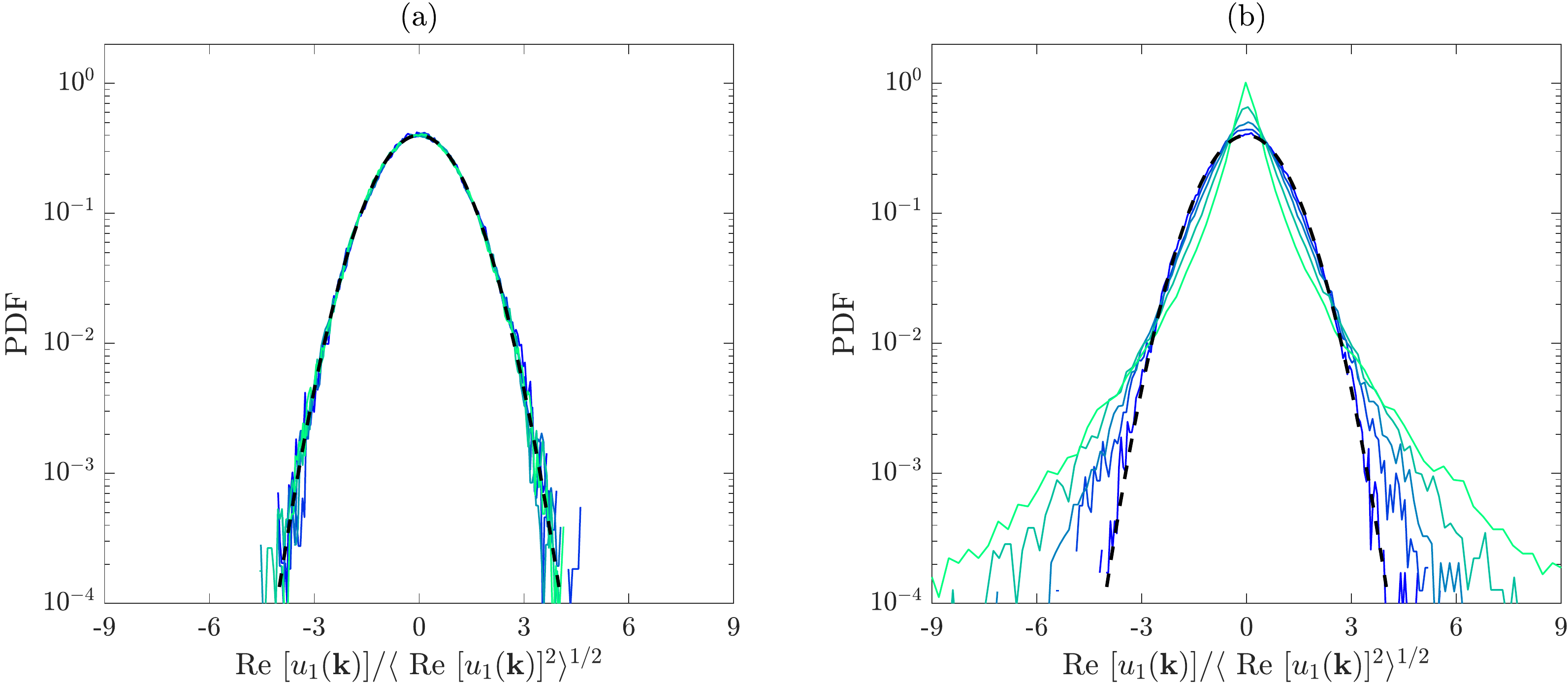}
		\caption{Normalized PDF's in logarithmic scale of the real part of $x$ velocity, $\text{Re} [u_1(\mathbf{k})]$, at different wave vectors $\mathbf{k}_n = \lambda^n \mathbf{k}_0$, $\mathbf{k_0} = (1,1,1)$, rescaled along the main diagonal of Fourier space. Scales decrease from darker to lighter colors. Gaussian distribution of zero mean and unit variance is shown for comparison in black dashed line. (a)~Statistics of inertial-range wave vectors $\mathbf{k}_n$, with $n = 10,\dots,15$; (b)~statistics of viscous-range wave vectors $\mathbf{k}_n$, with $n = 18,\dots,22$.}
		\label{FIG:PDF}
	\end{figure*}
	
	Fig.~\ref{FIG:PDF}(a) shows the statistics at inertial-range wave vectors $\mathbf{k}_n$, for $n = 10,\dots,15$.
	The PDF's for all scales are very close to a Gaussian distribution. %, with non systematic widening of tails.
	Similar Gaussian distributions for inertial-range Fourier components were observed for developed turbulence through full DNS~\cite{brun2001statistics} and laboratory experiments~\cite{petrossian1997sound,mouri2002probability,chevillard2005statistics}.
	For a flow of characteristic large scale $L$ and finite correlation length $\ell$ in physical space, the univariate statistics of Fourier modes in the inertial range are normally distributed in the asymptotic limit $\ell/L \to 0$, as a particular case of the Central Limit Theorem for weighted integrals~\cite{lumley1972statistical}.
	For these reasons, it is commonly argued that Fourier modes are not well suited for the study of extreme events that proportionate inertial-range turbulent intermittency.
	
	Large fluctuations in Fourier modes can only appear when viscous processes become important and initiate a complex interplay between nonlinearity and dissipation.
	In this regime, the velocity field exhibits strong intermittency, associated with spatial variation of large-scale motion rather than with intense small-scale structures~\cite{chen1993far}.
	Unlike what occurs in the inertial range, dissipative intermittency leaves fingerprints on viscous-range Fourier components, whose statistics develop widening of tails at smaller scales~\cite{brun2001statistics}.
	Such behavior is also reproduced by the logarithmic model.
	Fig.~\ref{FIG:PDF}(b) shows the statistics at viscous-range wave vectors $\mathbf{k}_n$, for $n = 18,\dots,22$, where we observe an increasing deviation from Gaussian distribution as we move towards finer scales of the flow.
	
	As presented above, there is a strong similarity between statistics of lattice variables from the logarithmic model and Fourier components of the full Navier-Stokes equations; for instance, compare Figs.~\ref{FIG:PDF}(a,b) of the present paper with Figs.~1(f,b) from the DNS results in~\cite{brun2001statistics}.
	However, it is quite intriguing that the Gaussian behavior in our model is in sharp contrast with statistics of other simplified models, which usually present some degree of inertial-range intermittency.
	We turn now to a brief discussion about their statistical behavior.
	Shell models of turbulence exhibit chaotic intermittent dynamics in the inertial interval with statistical properties close to the Navier-Stokes developed turbulence~\cite{gledzer1973system,ohkitani1989temporal,l1998improved}.
	On the other hand, the reduced wave vector set approximation (REWA) model displays only weak intermittency~\cite{eggers1991does,grossmann1996developed}. 
	A possible explanation for this feature was given in~\cite{brun2001statistics}, where it is argued that REWA model can be written in a spherical model framework~\cite{mou1993spherical} consisting of $N$ interacting subsystems each one describing the evolution of a velocity component in a certain direction.
	In this framework, modes should have Gaussian statistics~\cite{eyink1994large} and anomalous fluctuations would be destroyed in the limit $N \to \infty$~\cite{pierotti1997intermittency}.
	A tendency towards less intermittent regime when increasing the couplings is also observed in the tree models of turbulence~\cite{benzi19971}.
	In view of these results, is reasonable to relate the non-intermittent Fourier modes in our model to its rich triad couplings, although no rigorous conclusions can be made.
	How much Fourier decimation decreases intermittency in physical space is also not clear.
	This was observed for the Burgers equation with random decimation~\cite{buzzicotti2016intermittency} and for the Navier-Stokes equations, decimating from full to REWA model~\cite{grossmann1996developed}, but not for Sabra model~\cite{mailybaev2015continuous}, which retains turbulent intermittent dynamics in physical space.
	
	We repeat that the absence of anomalous fluctuations in individual Fourier modes does not mean lack of intermittency in the flow, since this is exactly the scenario for developed turbulence, and because the intermittency is seen in the same way within the dissipative range.
	To determine whether our model mimics physical-space intermittency or not, it would be necessary to probe it properly.
	The challenging question is precisely how to capture intermittency fingerprints on Fourier variables~\cite{brun2001statistics}, the only available quantities for our model. We also remark that a proper definition of integral quantities, like structure functions, should take into account lattice volumes, because these volumes vary considerably among different cells. For example, a straightforward way to introduce structure functions would be through powers of the energy flux, where all terms are properly balanced: $S_p(k) = \langle |\Pi_k|^{p/3} \rangle$. Using these definitions, our numerical simulations confirm the exact Kolmogorov scaling of structure functions in the inertial range, while the location of the dissipative range depends on $p$ in agreement with the dissipative intermittency of Fig.~\ref{FIG:PDF}(b). One faces the same subtlety when introducing a proper analogue for the Kolmogorov energy spectrum. We leave the more detailed analysis of these interesting but non-trivial questions to future work.
	
	\section{Conclusions}
	
	We propose a new strategy for constructing simplified models of fluid dynamics, which restricts the governing equations in their original form to a multi-dimensional logarithmic lattice in Fourier space.
	This domain receives a specially designed operational structure, which retains most of the usual calculus and algebraic properties.
	As a consequence, the resulting models preserve all symmetries (some in discrete form, namely scaling invariance and isotropy), inviscid invariants (energy and helicity, for 3D flow; energy and enstrophy, for 2D flow), and also reproduces some fine properties of Euler flow, like incompressibility and Kelvin's Circulation Theorem.
	The classification of all possible lattices supporting this construction allows us to obtain different dynamical models sharing all the above properties, and so to test the robustness and universality of the results they provide.
	Because of the strongly decimated domain, the logarithmic models can be easily simulated with great accuracy and covering a large spatial range.
	Furthermore, the solutions correlate with existing DNS at the correspondent scales~\cite{campolina2018chaotic}.
	
	After showing rigorously that the properties of plausible finite-time singularities (blowup) for the incompressible 3D Euler equations have similar form on the logarithmic lattice, we presented the numerical evidence of blowup, characterized as a chaotic wave in a renormalized system.
	Surprisingly similar asymptotic behavior of solutions was observed for two very different lattice models, probing the robustness of our conclusions, also drawn earlier in~\cite{campolina2018chaotic}.
	The multi-scale character of the attractor (ranging six decades in Fourier space) reveals the great complexity of the blowup and explains why there is a controversy around the available numerical studies, since actual computational techniques may be insufficient by far for the required resolution.
	Still, one may think of accessing the blowup through experimental measurements~\cite{saw2016experimental,kuzzay2017new,debue2018dissipation}.
	
	The viscous incompressible model on a logarithmic lattice exhibits anomalous dissipation in the limit of large Reynolds numbers, similarly to hydrodynamics turbulence~\cite{frisch1999turbulence}.
	This was demonstrated by measuring the mean energy flux in the inertial range for a sequence of decreasing viscosities.
	Moreover, statistics of lattice variables behave like Fourier components in the full Navier-Stokes turbulence, whose distributions are Gaussian in the inertial interval and intermittent at viscous scales.
	Such behavior contrasts with other simplified models, which usually display some degree of inertial-range intermittency.
	Though the question whether our logarithmic model reproduces a kind of physical-space intermittency was left open.
	We believe that future analysis of this model may help in better understanding the relation between physical and Fourier space representations in developed turbulence~\cite{brun2001statistics}.
	
	The systematic technique we presented is applicable to any partial differential equation with quadratic nonlinearities.
	In this framework, symmetries and quadratic invariants are expected to be automatically preserved due to the designed functional structure on the lattice.
	This turns the logarithmic models into a general methodology for the study of singularities and regularity in many nonlinear systems.
	We also developed the library \textsc{LogLatt}~\cite{campolina2020loglattmatlab}, an efficient \textsc{Matlab}\textsuperscript{\circledR} computational tool for the numerical calculus on logarithmic lattices.
	The proposed technique is ready-to-use in other fluid models, like natural convection~\cite{majda2002vorticity}, geostrophic motion~\cite{pedlosky2013geophysical,constantin1994formation}, porous media~\cite{cordoba2007analytical} and magnetohydrodynamics~\cite{biskamp1997nonlinear}.
	The lower computational cost of the logarithmic models compared to full DNS may be of great use for problems in higher dimensions, like high-dimensional turbulence~\cite{gotoh2007statistical,miyazaki2010classical,suzuki2005energy,yamamoto2012local}. With further extensions, there is a hope to apply this technique to compressible turbulence~\cite{augier2019shallow,falkovich2017vortices} and superfluids~\cite{biferale2019superfluid,biferale2019strong}; see the example in Appendix~\hyperref[app:C]{C} of a possible way to model isentropic compressible flow on logarithmic lattices.
	
	\section*{Aknowledgments}
	The authors thank L. Biferale, B. Dubrulle, U. Frisch, and S. Thalabard for useful discussions and suggestions. The work was supported by  CNPq (grants 303047/2018-6, 406431/2018-3).
	
	\section*{Appendix A: Burgers' representation for shell models}\label{app:A}
	
	In this Appendix, we show that some well-known shell models of turbulence are equivalent to the  Burgers equation on a logarithmic lattice.
	This, in particular, reinforces the idea that self-similar blowup and non-oscillatory Kolmogorov regime in shell models follow a scenario closer to Burgers' dynamics~\cite{mailybaev2012renormalization, mailybaev2015continuous} than to Euler's. 
	
	The Burgers equation~\cite{burgers1948mathematical} on the one-dimensional logarithmic lattice of spacing $\lambda$ is given by
	\begin{equation}
	\partial_t u + u \ast \partial_x u = \nu \partial_x^2 u,
	\label{EQ:Burgers}
	\end{equation}
	where $\nu \geq 0$ is the viscosity.
	First, let us take $\lambda = 2$ and consider the corresponding product~\eqref{EQ:product_1D_1} with a prefactor of $2$. The Burgers equation~\eqref{EQ:Burgers} takes the form
	\begin{equation}
	\partial_t u (k) = -ik\left[ 2 u(2 k)\overline{u(k)} + u^2 \left( \frac{k}{2} \right) \right] - \nu k^2 u(k).
	\label{EQ:DN_abstract}
	\end{equation}
	Define the geometric progression $k_n = \lambda^n, \ n \in \mathbb{Z}$ and consider purely imaginary solutions of type $u(\pm k_n) = \pm i u_n$ for $u_n \in \mathbb{R}$.
	Note that this is a property of the Fourier transform for any odd function in physical space. Then,  equation~\eqref{EQ:DN_abstract} taken at $k = k_n$ reduces to the form
	\begin{equation}
	\partial_t u_n = k_n u_{n-1}^2 - k_{n+1}u_{n+1} u_n - \nu k_n^2 u_n.
	\label{EQ:DN}
	\end{equation}
	This system is known as the Desnyansky-Novikov shell model~\cite{desnyansky1974evolution}, also called dyadic model.
	
	For our second example, we take $\lambda = \varphi$, the golden mean, and consider the  product \eqref{EQ:product_1D_1} with prefactor $-\varphi$.
	By setting $u(k_n) = u_n$ and $u(-k_n) = \overline{u_n}$ with $k_n = \varphi^n$, the Burgers equation~\eqref{EQ:Burgers} is reduced to the form
	\begin{equation}
	\partial_t u_n = i[k_{n+1}u_{n+2}\overline{u_{n+1}} - (1+c)k_n u_{n+1}\overline{u_{n-1}} - c k_{n-1}u_{n-1}u_{n-2}] - \nu k_n^2 u_n,
	\label{EQ:Sabra}
	\end{equation}
	with $c  = -\varphi^2$. System~\eqref{EQ:Sabra} is the Sabra shell model~\cite{l1998improved}.
	
	A third possibility is to consider $\lambda = \sigma$, the plastic number~\eqref{plastic}, which reduces Eq.~\eqref{EQ:Burgers} to a new shell model with improved number of triad interactions. In this spirit, extended triads were considered in the context of helical shell models~\cite{de2015inverse}.
	
	Model~\eqref{EQ:Burgers} on the logarithmic lattice retains several properties of the continuous Burgers equation, like the symmetries of time translation $t \mapsto t + t_0$ by any $t_0 \in \mathbb{R}$, physical-space translation $u(k) \mapsto e^{-ik\xi}u(k)$ by a number $\xi \in \mathbb{R}$ and, in the case of a lattice with origin, Galilean invariance $u(k,t) \mapsto e^{-ikvt}u(k,t) - \widehat{v}(k)$ for any $v \in \mathbb{R}$, where $\widehat{v}(0) = v$ and zero for $k \neq 0$.
	Inviscid ($\nu = 0$) regular solutions also conserve the momentum $M(t) = u(k = 0)$, energy $E(t) = \frac{1}{2}(u,u)$ and the thrid-order moment $M_3(t) = (u \ast u, u) = (u,u \ast u)$, which is well-defined because of associativity in average of the product -- see property~\ref{DEF:PROD_associativity_avg} in Definition~\ref{DEF:PROD}.
	All these conservation laws can be proved using only the operations on logarithmic lattices; see~\cite{campolina2019fluid}. 
	Conservation of energy is a well-known property of shell models while the conservation of a third-order moment was revealed in the study of Hamiltonian structure in Sabra model~\cite{lvov1999hamiltonian}. 
	Unlike the continuous Burgers equation, higher-order moments are not conserved for the logarithmic models because of non-associativity on the logarithmic lattice -- see Corollary~\ref{THE:product_associativity} -- which turns higher powers not even well-defined.
	The non-existence of invariants of order greater than $3$ was proved in~\cite{ditlevsen2000symmetries} for the Sabra model.
	Sabra model has one more inviscid quadratic invariant of the form $I = \sum_{n \in \mathbb{Z}} c^{-n} |u_n|^2$, but this invariant do not seem to have an analogue in the Burgers equation.
	In studies of hydrodynamic turbulence, it was interpreted as the enstrophy for $c>0$ (sign definite invariant) and as helicity for $c<0$ (not sign-definite invariant).
	
	Our methodology not only reproduces shell models but also leads to new insights about them.
	In the spirit of Theorem~\ref{THM:Kelvin}, consider a scalar field $\rho$ evolving as
	\begin{equation}
	\partial_t \rho + \partial_x (\rho \ast u) = 0.
	\end{equation}
	This equation mimics a passive scalar advected by the flow, e.g. density. Then, the cross-correlation
	\begin{equation}
	\Gamma(t) = (\rho,u)
	\end{equation}
	which can be seen as total momentum of the flow, is conserved in time; the proof follows similar lines as those already presented and may be found in~\cite{campolina2019fluid}.
	Since this conservation holds for all solutions $\rho(t)$, this provides infinitely many inviscid invariants for model~\eqref{EQ:Burgers}, analogous to circulation in Kelvin's Theorem as described in Section~\ref{SEC:ideal_flow}.
	Up to our knowledge, this has not been shown earlier.
	
	\section*{Appendix B: Functional inequalities on logarithmic lattices}\label{app:B}
	
	Here we prove some functional inequalities and operator properties used in Section~\ref{SEC:local_theory}.
	
	\begin{lemma}
		Let $\mathbf{u} \in h^m$ and $\mathbf{v} \in h^{m-1}$, for $m \geq 1$. Then, $\mathbf{u} \ast \mathbf{v} = \sum_{i = 1}^{d}u_i \ast v_i \in h^m$ with
		\begin{equation}
		||\mathbf{u} \ast \mathbf{v}||_{h^m} \leq C(||\mathbf{u}||_{h^m} ||\mathbf{v}||_{\ell^\infty} + ||D\mathbf{u}||_{\ell^\infty} ||\mathbf{v}||_{h^{m-1}}),
		\label{EQ:main_inequality}
		\end{equation}
		where $C$ is a constant which does not depend on $\mathbf{u}$ and $\mathbf{v}$.
	\end{lemma}
	\begin{proof}
		Let us prove the inequality in the one-dimensional case.
		Using elementary algebraic relations, we obtain
		\begin{align*}
		||u \ast v||_{h^m}^2 &= ||D^m(u \ast v)||_{\ell^2}^2
		= \sum_{k \in \mathbb{\Lambda}} |k|^{2m}|(u \ast v)(k)|^2 \\
		&\leq N\sum_{k \in \mathbb{\Lambda}} \sum_{j = 1}^{N}|k|^{2m}|u(p_jk)v(q_jk)|^2 \\
		&= N\sum_{k \in \mathbb{\Lambda}} \sum_{j = 1}^{N}|p_jk + q_jk|^{2m}|u(p_jk)v(q_jk)|^2 \\
		&\leq 2^{2m-1}N\sum_{k \in \mathbb{\Lambda}} \sum_{j = 1}^{N}(|p_jk|^{2m} + |q_jk|^{2m})|u(p_jk)v(q_jk)|^2 \\
		&= 2^{2m-1}N\sum_{k \in \mathbb{\Lambda}} \sum_{j = 1}^{N}|p_jk|^{2m}|u(p_jk)v(q_jk)|^2 +2^{2m-1}N \sum_{k \in \mathbb{\Lambda}} \sum_{j = 1}^{N}|q_jk|^{2m}|u(p_jk)v(q_jk)|^2.
		\end{align*}
		In the first term, we estimate
		\begin{equation*}
		\sum_{k \in \mathbb{\Lambda}} \sum_{j = 1}^{N}|p_jk|^{2m}|u(p_jk)v(q_jk)|^2 \leq ||v||_{\ell^\infty}^2 \sum_{j = 1}^{N}\sum_{k \in \mathbb{\Lambda}}|p_jk|^{2m}|u(p_jk)|^2
		\leq N||u||_{h^m}^2||v||_{\ell^\infty}^2,
		\end{equation*}
		while the sums of the second term are bounded by
		\begin{align*}
		\sum_{k \in \mathbb{\Lambda}} \sum_{j = 1}^{N}|q_jk|^{2m}|u(p_jk)v(q_jk)|^2 &= \sum_{k \in \mathbb{\Lambda}} \sum_{j = 1}^{N}|q_jk|^{2m-2}|q_jk|^{2}|u(p_jk)v(q_jk)|^2 \\
		&\leq M\sum_{k \in \mathbb{\Lambda}} \sum_{j = 1}^{N}|q_jk|^{2m-2}|p_jk|^{2}|u(p_jk)v(q_jk)|^2 \\
		&\leq M ||Du||_{\ell^\infty}^2 \sum_{j = 1}^{N} \sum_{k \in \mathbb{\Lambda}} |q_jk|^{2m-2}|v(q_jk)|^2 \\
		&\leq MN ||Du||_{\ell^\infty}^2 ||v||_{h^{m-1}}^2,
		\end{align*}
		where $M = \max_{j = 1,\dots,N} |q_j|^2/|p_j|^2$.
		In view of the estimates for the two terms, we reach to the result $|| u \ast v ||_{h^m} \le C\left(||u||_{h^m}||v||_{\ell^\infty} + ||Du||_{\ell^\infty}||v||_{h^{m-1}}\right)$ with the choice of $C = 2^{m-1/2}N\max(M,1)^{1/2}$. The proof extends naturally to higher dimensions, by considering multiple components.
	\end{proof}
	\begin{lemma}
		Define the bilinear operator
		\begin{equation}
		B(\mathbf{u},\mathbf{v}) = \mathbf{u} \ast \nabla \mathbf{v},
		\end{equation}
		where $(\mathbf{u} \ast \nabla \mathbf{v})_i = \mathbf{u} \ast \nabla u_i = \sum_{j = 1}^{d} u_j \ast \partial_j v_i$.
		Then, $B: h^m \times h^m \to h^m$ is a bounded bilinear operator, i.e., there exists a constant $C>0$ such that
		\begin{equation}
		||B(\mathbf{u},\mathbf{v}) ||_{h^m} \leq C ||\mathbf{u}||_{h^m} ||\mathbf{v}||_{h^m},
		\end{equation}
		for every $\mathbf{u},\mathbf{v} \in h^m$.
	\end{lemma}
	\begin{proof}
		Using inequality \eqref{EQ:main_inequality} for $\mathbf{u}$ and $\nabla v_i$, we obtain
		\begin{equation*}
		||B(\mathbf{u},\mathbf{v})||_{h^m} 
		\le \sum_{i = 1}^d ||\mathbf{u} \ast \nabla v_i||_{h^m} 
		\le C\sum_{i = 1}^d\left( ||\mathbf{u}||_{h^m}||\nabla v_i||_{\ell^\infty} + ||D\mathbf{u}||_{\ell^\infty}
		||\nabla v_i||_{h^{m-1}} \right).
		\end{equation*}
		We now use the inequalities
		\begin{equation*}
		||\nabla v_i||_{\ell^{\infty}} 
		\le ||D\mathbf{v}||_{\ell^\infty} 
		\le ||\mathbf{v}||_{h^1}, \quad 
		||D\mathbf{u}||_{\ell^\infty} \leq ||\mathbf{u}||_{h^1}, \quad 
		||\nabla v_i||_{h^{m-1}} \le ||\mathbf{v}||_{h^m}
		\end{equation*}
		and the general relation
		\begin{equation*}
		||\mathbf{u}||_{h^1} \le ||\mathbf{u}||_{h^m},
		\end{equation*}
		which are simple estimates from the definition of the norms on the lattice (\ref{EQ:Euler_lattice}).
		This yields 
		\begin{equation*}
		||B(\mathbf{u},\mathbf{v})||_{h^m} \le 2dC||\mathbf{u}||_{h^m} ||\mathbf{v}||_{h^m},
		\end{equation*}
		which shows that $B(\mathbf{u},\mathbf{v}) \in h^m$ and the boundness of operator $B$.
	\end{proof}
	
	\section*{Appendix C: Isentropic compressible flow}\label{app:C}
	
	The logarithmic models presented in this paper do not extend naturally to isentropic (or general) compressible flow due to the appearance of cubic terms in the governing equations and inviscid invariants.
	Nevertheless, we present below one possible way to overcome this issue.
	The idea consists of introducing additional variables properly constrained, so the original cubic terms become quadratic with respect to the extended set of variables.
	In this formulation, the symmetries and conserved quantities are exactly those from the continuous model. Unfortunately, preliminary numerical simulations do not show good correspondence to dynamical features of realistic compressible flows, such as formation of shock waves. 
	For this reason we restrict ourselves to the model description and its conserved quantities, leaving the numerical implementation for future analysis.
	
	\textsc{Model.} We introduce the scalar density $\rho(\mathbf{k},t)$, the velocity field $\mathbf{u}(\mathbf{k},t)$ and the momentum field $\mathbf{q}(\mathbf{k},t)$, defined on the lattice $\mathbf{k} \in \mathbb{\Lambda}^d$.
	The model for ideal compressible flow consists of the continuity equation and the balance of momentum together with an algebraic constraint relating all variables, respectively given by the system (cf.~\cite[Sec.~15]{landau1987fluid})
	\begin{equation}
	\partial_t \rho + \nabla \cdot \mathbf{q} = 0,\quad
	\partial_t \mathbf{q} = - \nabla \cdot \Pi,\quad
	\mathbf{q} = \rho \ast \mathbf{u},
	\label{EQ:isentropic}
	\end{equation}
	where the \textit{momentum flux density} tensor $\Pi$ has its classical form
	\begin{equation}
	\Pi_{ij} = p \delta_{ij} + u_i \ast q_j.
	\end{equation}
	In an isentropic flow, the pressure $p$ is a function of the density.
	For our logarithmic model, we consider the quadratic relation
	\begin{equation}
	p = A \rho \ast \rho,
	\label{EQ:polytropic}
	\end{equation}
	which mimics a polytropic gas $p = A \rho^\gamma$, with $\gamma = 2$.
	
	To evolve model~\eqref{EQ:isentropic}, one needs to solve the last algebraic constraint for the velocities $\mathbf{u}$, i.e., express it in terms of the momentum and density.
	This is possible when the mean density $\rho(\mathbf{0})>0$ at $\mathbf{k} = \mathbf{0}$ is sufficiently larger than the sum of all other components $\sum_{\mathbf{k}\neq \mathbf{0}}|\rho(\mathbf{k})|$.
	Under this condition, the density field may be interpreted as small-amplitude oscillations around a positive mean value.
	Solvability of velocities under this condition can be rigorously proved in proper functional spaces using Operator Theory.
	
	\textsc{Conserved quantities.} The \textit{total momentum} of the flow is naturally defined as
	\begin{equation}
	M(t) = \mathbf{q}(\mathbf{0},t),
	\label{EQ:isentropic_momentum}
	\end{equation}
	at $\mathbf{k} = \mathbf{0}$.
	The \textit{total energy} $E$ decomposes into two contributions
	\begin{equation}
	E = K + U,
	\label{EQ:isentropic_total_energy}
	\end{equation}
	where
	\begin{equation}
	K = \frac{1}{2}(\mathbf{q},\mathbf{u})
	\end{equation}
	is the \textit{kinetic energy} and
	\begin{equation}
	U = (\rho,e)
	\end{equation}
	is the \textit{internal energy}.
	The \textit{internal energy per unit mass} $e$ is defined as
	\begin{equation}
	e = A\rho.
	\label{EQ:internal_energy_mass}
	\end{equation}
	Formula~\eqref{EQ:internal_energy_mass} is obtained from the pressure through the well-known (isentropic) thermodynamical relation $de = pd\rho/\rho^2$.
	
	System~\eqref{EQ:isentropic} conserves total momentum~\eqref{EQ:isentropic_momentum} and total energy~\eqref{EQ:isentropic_total_energy} in time.
	Kinetic and internal energies are transferred from one another through pressure as
	\begin{equation}
	\frac{dK}{dt} = -\frac{dU}{dt} = -(\nabla p, \mathbf{u}).
	\end{equation}
	
	\textsc{Viscous effects.} Following classical derivations of fluid mechanics, viscosity is introduced in the momentum flux density tensor as
	\begin{equation}
	\Pi_{ij} = p \delta_{ij} + u_i \ast q_j - \sigma_{ij},
	\end{equation}
	with the \textit{viscous tensor} $\sigma$ given by
	\begin{equation}
	\sigma_{ij} = \eta \left( \partial_i u_j + \partial_j u_i - \frac{2}{3}\nabla \cdot \mathbf{u}\delta_{ij} \right) + \zeta \nabla \cdot \mathbf{u}\delta_{ij}.
	\end{equation}
	The constants $\eta, \zeta \geq0$ are the \textit{viscosity coefficients}.
	Non-equilibrium solutions dissipate energy through the work of viscosity forces in the form
	\begin{equation}
	\frac{dE}{dt} = (\nabla \cdot \sigma, \mathbf{u}).
	\end{equation}
	In this way, system (\ref{EQ:isentropic}) yields equations for the viscous flow.
	
	\bibliographystyle{plain}
	\bibliography{refs}
\end{document}